\documentclass[aps,prx,reprint,floatfix,amsmath,amssymb,nofootinbib,bibnotes,longbibliography,superscriptaddress]{revtex4-2}

\usepackage{tablefootnote}
\usepackage{array}[=2016-10-06]
\usepackage[T1]{fontenc}
\makeatletter
\providecommand{\tabularnewline}{\\}
\makeatother
\usepackage{varwidth}
\usepackage{stmaryrd}
\usepackage{graphicx} 	
\usepackage{tikz}
\usepackage{calc}

\usepackage{amsthm}
\usepackage{amssymb}
\usepackage{amsmath}
\usepackage{enumitem}
\usepackage{tabularx}
\usepackage{algorithm}
\usepackage{algpseudocode}
\usepackage{DejaVuSans}
\usepackage{sansmath}

\newtheorem{proposition}{Proposition}

\newtheorem{example}{Example}

\newtheorem{observation}{Observation}

\usepackage[acronym]{glossaries}
\usepackage{xurl}
\usepackage{hyperref}
\usepackage{titletoc}

\usepackage[T1]{fontenc}
\usepackage{lmodern}
\usetikzlibrary{arrows.meta,calc,positioning,fit,decorations.pathreplacing}

\newsavebox{\memorypanelbox}
\newsavebox{\resultspanelbox}

\usepackage{hyperref}
\hypersetup{
    colorlinks=true,
    linkcolor=blue,
    filecolor=blue,
    citecolor=blue,  
    urlcolor=blue,
    linktocpage=true
}

\makeglossaries

\glsdisablehyper

\algrenewcommand\algorithmicrequire{\textbf{Input:}}
\algrenewcommand\algorithmicensure{\textbf{Output:}}
\algnewcommand\Input{\item[\algorithmicrequire]}
\algnewcommand\Output{\item[\algorithmicensure]}
\algrenewcommand\algorithmiccomment[1]{\hfill{\color{gray}\texttt{\textbackslash\textbackslash}~#1}}
\algnewcommand\LComment[1]{{\color{gray}\texttt{\textbackslash*}~#1~\texttt{*\textbackslash}}}

\newcommand{\E}[1]{\langle #1 \rangle}
\newcommand{\bsk}{\boldsymbol{k}}

\newcommand{\bsph}{\boldsymbol{\phi}}

\newcommand{\bbra}[1]{\langle\!\langle #1|}
\newcommand{\kket}[1]{|#1\rangle\!\rangle}
\newcommand{\ket}[1]{|#1\rangle}

\newcommand{\V}{\mathbb{V}}

\newcommand{\bst}{\boldsymbol{\theta}}
\newcommand{\me}{\mathsf{e}}

\newcommand{\tr}{\operatorname{Tr}}

\newacronym{p2lc}{P2LC}{physical-to-logical characterization}
\newacronym{zne}{ZNE}{Zero-Noise Extrapolation}
\newacronym{lep}{LEP}{logical error probability}
\newacronym{lem}{LEM}{\emph{logical error mitigation}}
\newacronym{pea}{PEA}{Pauli error amplification}
\newacronym{qp}{QP}{quasi-probability}
\newacronym{em}{EM}{error mitigation}
\newacronym{ec}{EC}{error corrected}
\newacronym{av}{AV}{\emph{active volume}}
\newacronym{ivw}{IVW}{inverse variance weighting}

\newacronym{mname}{SOLEM}{\emph{syndrome-based offline} logical error mitigation}

\DeclareMathOperator*{\argmin}{arg\,min}

\usepackage[english]{babel}
\usepackage[dvipsnames,svgnames,x11names]{xcolor}

\usepackage{etoolbox}

\makeatletter
\AtBeginDocument{%
  \patchcmd{\NAT@sort@cites}
    {\expandafter\immediate\expandafter\write}
    {\expandafter\write}
    {}
    {\PackageError{citeorder}
      {Citation-order patch failed}
      {Check the natbib version}}%
}
\makeatother

\begin{document}

\title{
Error Mitigation From a Single Unmodified Error-Corrected Circuit
}

\begin{abstract}

Recent work has shown that syndrome information obtained during error-corrected quantum computation can be utilized for \textit{logical error mitigation} (LEM), yielding exponential reductions in runtime overhead compared to ``external'' LEM methods, which do not make use of syndrome data. 
However, these \textit{syndrome-aware} LEM methods (SALEM) generally require adaptive circuit operations based on real-time processing of syndromes.
Here, we introduce \textit{syndrome-based offline} logical error mitigation (SOLEM), 
a framework that utilizes syndrome data entirely in post-processing, without \textit{any} modification to the given error-corrected circuit. 
In this sense, SOLEM is ``built into'' error-corrected circuits, and
is in fact simpler to implement than external LEM protocols. 
The basic idea of SOLEM is to use naturally occurring syndrome-conditioned logical error channels in place of deliberate circuit modifications. We develop unbiased and heuristic estimators within this framework, including offline logical versions of probabilistic error cancellation (PEC) and zero-noise extrapolation (ZNE), 
as well as a global-regression estimator.

SOLEM's fully offline nature leads to three significant advantages.
First, SOLEM can post-process syndromes with a more powerful decoder than the one used in real time.
Second, multiple estimators can be constructed from the same measured data and optimized using empirical statistical errors and related performance metrics.
Third, logical error characterization is needed only in post-processing, allowing characterization data to be collected during circuit execution and reducing sensitivity to hardware drift.
Remarkably, SOLEM estimators can achieve these advantages while attaining a runtime overhead which is exponentially lower than that of external LEM and is comparable to that of corresponding SALEM estimators.

\end{abstract}

\author{Maor Shutman}
\email{maor.shutman@qedma.com}
\affiliation{Qedma Quantum Computing, Tel Aviv, Israel}

\author{Dorit Aharonov}
\affiliation{Qedma Quantum Computing, Tel Aviv, Israel}
\affiliation{The Benin School of Computer Science and Engineering, Hebrew University, Jerusalem, Israel}

\author{Itai Arad}
\affiliation{Qedma Quantum Computing, Tel Aviv, Israel}

\author{Yosi Atia}
\affiliation{Qedma Quantum Computing, Tel Aviv, Israel}

\author{Eyal Bairey}
\affiliation{Qedma Quantum Computing, Tel Aviv, Israel}

\author{Zvika Brakerski}
\affiliation{Qedma Quantum Computing, Tel Aviv, Israel}
\affiliation{Faculty of Mathematics and Computer Science,
Weizmann Institute of Science, Israel}

\author{Barak A.~Katzir}
\affiliation{Qedma Quantum Computing, Tel Aviv, Israel}

\author{Netanel H.~Lindner}
\affiliation{Qedma Quantum Computing, Tel Aviv, Israel}
\affiliation{Department of Physics, Technion, Haifa, Israel}

\author{Omri Golan}
\email{omri.golan@qedma.com}
\affiliation{Qedma Quantum Computing, Tel Aviv, Israel}

\author{Ilya Gurwich}
\affiliation{Qedma Quantum Computing, Tel Aviv, Israel}

\maketitle

\begin{figure*}[t]
    \centering
    \resizebox{\textwidth}{!}{\input{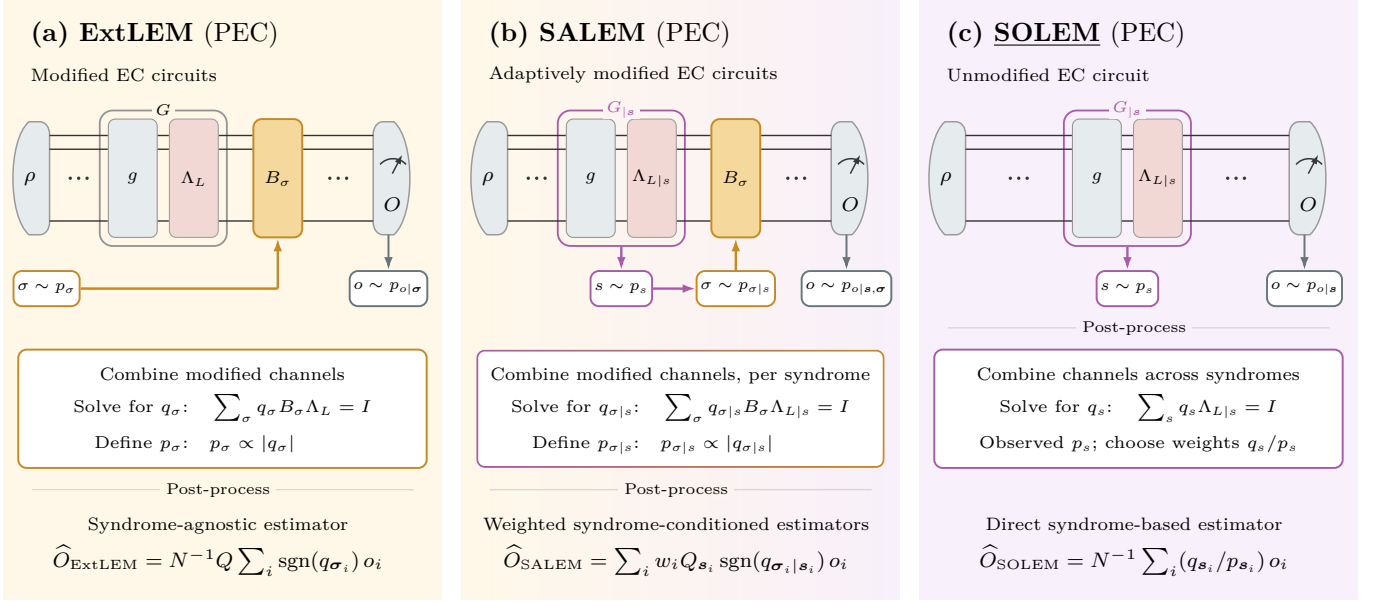}}
    \caption{\textbf{Conceptual overview of SOLEM.}
We compare \textit{syndrome-based offline} logical error mitigation (SOLEM), introduced here, with the previously proposed \textit{external} LEM (ExtLEM) and \textit{syndrome-aware} LEM (SALEM). As an example, logical versions of \textit{probabilistic error cancellation} (PEC) are shown. In an error-corrected circuit (EC circuit), each faulty layer \(G\) can be described as an ideal logical layer \(g\) followed by a logical error channel \(\Lambda_L\) (see Sec.~\ref{sec:setup}). Conditioned on the syndrome $s$ measured within $G$ (and, implicitly, also in past layers) the relevant channel is \(\Lambda_{L\mid s}\), occurring with probability \(p_s\). Purple indicates syndromes $s$ and their usage, while yellow indicates circuit modifications, labeled by $\sigma$, and their classical preparation. Bold symbols $\boldsymbol{s},\boldsymbol{\sigma}$ indicate full-circuit records of $s,\sigma$. 
\textbf{(a)} ExtLEM applies a conventional EM method targeting the channel \(\Lambda_L\). For PEC, one solves in pre-processing for a quasi-probability (QP) distribution \(q_\sigma\) over circuit modifications \(B_\sigma\) that effectively replaces \(\Lambda_L\) by the identity channel \(I\); then samples modifications from an optimized probability distribution \(p_\sigma\propto |q_\sigma|\). In post-processing, outcomes are averaged using the QP signs and normalization \(Q=\lVert q\rVert_1\), producing an unbiased estimator \(\widehat{O}_{\rm ExtLEM}\) for the ideal logical expectation value \(\langle O\rangle \), with shot overhead $\leq Q^2$.
{\bf (b)} In SALEM \cite{aharonov2025syndromeawaremitigationlogical}, mitigation is conditioned on the syndrome \(s\) measured during execution, targeting the syndrome-conditioned channel \(\Lambda_{L\mid s}\). This involves conditioned QP distributions \(q_{\sigma\mid s}\) and corresponding distributions \(p_{\sigma\mid s}\propto |q_{\sigma|s}|\), from which modifications \(B_\sigma\) are sampled adaptively. Outcomes are combined using the conditional QP signs and norms
$Q_{\boldsymbol{s}}$, as well as  inverse-variance
weights $w_i\propto Q_{\boldsymbol{s}_i}^{-2}$. The resulting estimator \(\widehat{O}_{\rm SALEM}\) generically has an exponentially lower shot overhead relative to \(\widehat{O}_{\rm ExtLEM}\).
{\bf (c)} SOLEM executes the given EC circuit without \textit{any} modifications \(B_\sigma\), requiring no added classical processing during or before circuit execution. In post-processing, the measured syndrome and observable records are combined, using a QP distribution \(q_s\) which (ideally) reproduces the identity channel directly from the naturally occurring channels \(\Lambda_{L\mid s}\). The latter play the role of modified channels $B_\sigma\Lambda_L$ in ExtLEM. Although the natural syndrome distribution \(p_s\) is not optimized to be \(\propto |q_s|\), the SOLEM estimator \(\widehat{O}_{\rm SOLEM}\) can admit an exponentially reduced shot overhead relative to ExtLEM, comparable to that of SALEM. 
}
    \label{fig:overview_schematics}
\end{figure*}

\section{Introduction and summary \label{Sec: intro}}
Quantum error correction (EC) theory and experiment are rapidly
advancing~\cite{bravyi2024high,zhou2025transversal,acharya2025below,
bluvstein2026fault,dasu2026computing}, and hardware roadmaps target
useful EC by the end of this
decade~\cite{ibm2026roadmap,quantinuum2025apollo}. However, various hardware limitations, and in particular limited numbers of physical qubits, are expected to restrict the performance of EC, leaving non-negligible logical errors that constrain the achievable volume and accuracy of quantum computations in the foreseeable future. 

Error mitigation (EM) protocols can extend
the reach of \textit{non}-error-corrected quantum processors
\cite{Temme2017QEM,Endo2018QEM,Cai2023QEM,Aharonov2025Importance}, as now demonstrated extensively
\cite{Kim2023Utility,aharonov2025reliablehighaccuracyerrormitigation,Abanin2025Constructive,Haghshenas2026DigitalMagnetism,Fischer2026QuantumChaos,Shinjo2026TimeCrystals,Alam2025Programmable,Alam2025Fermionic,Granet2025Pairing,Hartnett2026FermiHubbard,Leviatan2026PrethermalFloquet,Barron2026ObservableEstimation}. Motivated by this line of work,  \textit{logical} error mitigation (LEM) schemes were proposed 
\cite{Suzuki2022, Piveteau2021, Lostaglio2021, Xiong2020, Tsubouchi2024,Wahl2023DistanceZNE, aharonov2025syndromeawaremitigationlogical, Zhou2025ErrorMitigation,dincă2025errormitigationlogicalcircuits,tsubouchi2026quantumadvantagessyndromeawarenoisy,zhong2025combiningerrordetectionmitigation,yuan2026,kumar2026,Umbrarescu2026InfiniteDistance,Zhang2026LogicalQEM,Jeon2026PhysicalQEM, fischer2026spacetimemitigationlogicalerrors}, aiming to eliminate the effect of residual logical errors in error-corrected (or detected) logical circuits, in exchange for a shot overhead. The latter manifests as a quantum runtime overhead, which can be traded to reduce the space overhead (and intrinsic time overhead) of EC.

The early proposals for LEM protocols are \textit{external} (ExtLEM); they apply standard EM protocols to error-corrected logical gates, treating these as black boxes without exploiting their internal syndrome data for mitigation \cite{Suzuki2022, Piveteau2021, Lostaglio2021, Xiong2020, Tsubouchi2024}. Such ExtLEM protocols generally execute a set of modified EC circuits, whose final logical measurements are post-processed to estimate an ideal (error-free) logical expectation value. Fig.~\ref{fig:overview_schematics}(a) shows a prototypical example, where \textit{probabilistic error cancellation} (PEC) \cite{Temme2017QEM, Endo2018QEM} is applied externally to an EC circuit. In essence, PEC represents ideal logical layers using quasi-probability distributions (QP) over accessible modifications of faulty EC layers (see figure caption for details).   

Recently, Refs.~\cite{aharonov2025syndromeawaremitigationlogical, Zhou2025ErrorMitigation,dincă2025errormitigationlogicalcircuits, tsubouchi2026quantumadvantagessyndromeawarenoisy} showed that the syndrome data generated during EC can serve as a valuable resource for LEM; reusing it can reduce the exponential `blowup' rate with which the shot overhead grows with logical error rate and circuit volume. In particular, Refs.~\cite{aharonov2025syndromeawaremitigationlogical,  tsubouchi2026quantumadvantagessyndromeawarenoisy} studied \textit{syndrome-aware} LEM methods (SALEM) aimed at minimizing the shot overhead given syndrome data, finding that this can require \textit{adaptive} circuit modifications, extending the decoding and recovery process of EC. Fig.~\ref{fig:overview_schematics}(b) shows the shot-optimal PEC-based protocol of Ref.~\cite{aharonov2025syndromeawaremitigationlogical}, in which the circuit modifications of PEC become adaptive. The required real-time processing can be reduced by `coarse-graining' syndromes into a few classes, but generally at the cost of an increase in shot overhead. 

\begin{figure*}[t]
  \centering
  \includegraphics[width=\textwidth]{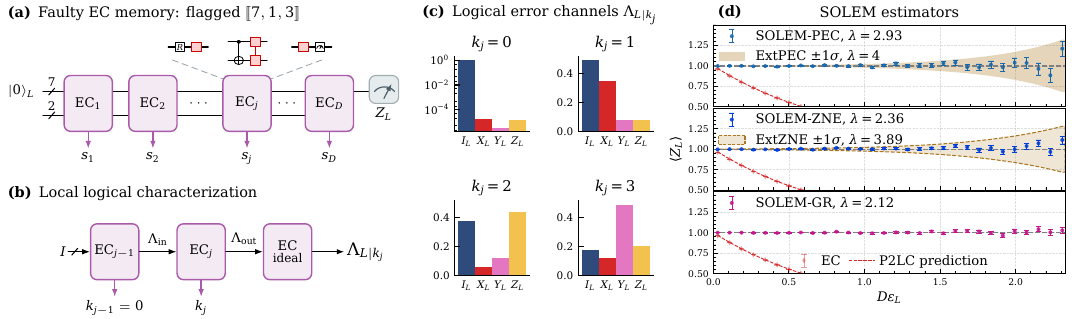} 
  \caption{ 
    \textbf{Circuit-level simulations of SOLEM.} 
    \textbf{(a)} We consider deep logical memory circuits based on the Steane code, \([\![7,1,3]\!]\). $D\gg1$ flagged EC cycles \cite{Chao_2018} are applied to the initial logical state \(\ket{0}_L\), followed by a measurement of 
    \(Z_L\). Each cycle \(\mathrm{EC}_j\) is subjected to a circuit-level noise model with physical error rate \(\epsilon=4\times10^{-4}\).
    Every shot produces a syndrome record \(\boldsymbol{s}\) and a final measurement  outcome \(o\in\{-1,1\}\). 
    \textbf{(b)} We use physical-to-logical characterization (P2LC) \cite{aharonov2025syndromeawaremitigationlogical} to map the given circuit-level noise model to \textit{local} logical error channels  \(\Lambda_{L|k_{j}}\), associated with individual EC cycles \(\mathrm{EC}_j\), and conditioned on a syndrome-class index $k_j\in\{0,1,2,3\}$. For simplicity, we suppress here the dependence of \(\Lambda_{L|k_j} \) on the past index $k_{j-1}$, which is fixed to $0$ (see Sec.~\ref{subsec:logical error ch} and \ref{sec:circ_leve_simul}).  
    \textbf{(c)} The four syndrome classes are 
    obtained 
    by separating a low-error subset ($k_j=0$) and clustering remaining syndromes according to the non-identity coefficients of their syndrome-conditioned logical channels. The resulting channels \(\Lambda_{L|k_j} \) are naturally occurring analogs of the Pauli-modified channels $I_L\Lambda_{L},X_L\Lambda_{L},Y_L\Lambda_{L},Z_L\Lambda_{L}$, used in external logical PEC and ZNE (ExtPEC, ExtZNE).   
    The syndrome-averaged logical error rate is \(\epsilon_L=1.12\times 10^{-4}\).
    \textbf{(d)} We implement the PEC,
    ZNE, and GR variants of \gls{mname} using the \gls{p2lc} channels and measured $o,\boldsymbol{s}$. The unmitigated decay of \(\E{Z_L}\) is accurately predicted by the \gls{p2lc} channels,
    and all three SOLEM estimators are statistically
    consistent with the ideal value of 1. 
    The legends report the exponential `blowup-rate'
    \(\lambda=(D\epsilon_L)^{-1}\log( \Gamma / \Gamma_{\rm {EC} })\) of the shot overhead \(\Gamma\), for the three \gls{mname} estimators, at the largest
    simulated depths \(D=2.3\epsilon_L^{-1}\).
    At the same shot count of \(N= 10^5\), the shaded region shows the predicted \(1\sigma\) intervals for ExtPEC and ExtZNE.
    The value \(\lambda=2.93\) for \gls{mname}-PEC demonstrates that, with the chosen partition, the naturally occurring channels and class probabilities $p_{k_j}$ enable a more efficient mitigation than the Pauli-modified channels of ExtLEM and their optimized probabilities, which yield $\lambda=4$. A similar conclusion  follows by comparing SOLEM-ZNE to ExtZNE, whose exponential heuristic holds exactly in this circuit. 
    Though the underlying heuristic of SOLEM-GR does not hold exactly in this circuit, 
    the GR estimator is statistically consistent with the ideal value, and attains exponentially lower sampling overhead than its ZNE and PEC counterparts, with $\lambda=2.12$. 
    }
  \label{fig:memory-simulation}
\end{figure*}

\begin{table*}[t]
\centering
\caption{
Leading-order sampling overheads of SOLEM variants and existing LEM
methods in the solvable model of Sec.~\ref{sec:model}.
The model has \(V\) statistically independent sites, with independent
syndrome labels \(k\in\{0,1\}\), \(p_1=\Pr(k=1)\), and parallel
syndrome-conditioned channels
\(\Lambda_{L|k}=(1-\epsilon_k)I+\epsilon_kE\).
The mean logical error probability is
\(\epsilon_L=(1-p_1)\epsilon_0+p_1\epsilon_1\).
We assume the exponential response of Prop.~\ref{prop:exp_Cr},
with error-dilution parameter \(\rho\), \(o\in\{-1,1\}\),
in the regime
\(0<\epsilon_0\ll\epsilon_1<1\) and \(p_1\ll1\).
Here \(p_{1|L}=p_1\epsilon_1/\epsilon_L\) and
\(\bar{\lambda}=2-\rho\epsilon_1p_{1|L}\).
The term \(\E{O}^2\) is to be subtracted from each listed expression.
}
\label{tab:overhead-summary}

\begingroup
\renewcommand{\arraystretch}{1.05}
\setlength{\tabcolsep}{3pt}

\newcommand{\formula}[1]{%
  \parbox[t]{0.292\textwidth}{%
    \centering
    \(\displaystyle #1\)\par}}

\newcommand{\tabnote}[1]{%
  \parbox[t]{0.292\textwidth}{%
    \centering\scriptsize #1\par}}

\newcommand{\methodcell}[1]{%
  \parbox[t]{0.085\textwidth}{%
    \raggedright #1\par}}

\newcommand{\headercell}[1]{%
  \parbox[c]{0.292\textwidth}{%
    \centering\textbf{#1}\par}}

\begin{tabular}{@{}llll@{}}
\hline\hline
& \headercell{PEC}
& \headercell{ZNE}
& \headercell{GR}
\tabularnewline
\hline
\noalign{\vskip 4pt}

\methodcell{\textbf{ExtLEM}\\[-2pt]
  {\scriptsize Ref.~\cite{filippov2024scalability}}}
& \formula{
    e^{4V\epsilon_L}
  }
& \formula{
    e^{2\rho V\epsilon_L}
    (1+3.6\rho V\epsilon_L)^2
  }
& \formula{\text{\textemdash}}
\tabularnewline[0pt]

&
& \tabnote{
    Amplification \(a=1+1.3/(\rho V\epsilon_L)\).
  }
& \tabnote{Not considered.}
\tabularnewline
\noalign{\vskip 4pt}
\hline
\noalign{\vskip 4pt}

\methodcell{\textbf{SALEM}\\[-2pt]
  {\scriptsize
    Ref.~\cite{aharonov2025syndromeawaremitigationlogical}}}
& \formula{
    e^{[4(1-p_{1|L})
       +(1-Q^{-2})\epsilon_1^{-1}p_{1|L}]V\epsilon_L}
  }
& \formula{
    e^{(p_{1|L}\epsilon_1^{-1}
       +2\rho\epsilon_0\epsilon_L^{-1})V\epsilon_L}
    (1+3.6\rho V\epsilon_0)^2
  }
& \formula{\text{\textemdash}}
\tabularnewline[0pt]

& \tabnote{
    Bit-flip: \(Q=|1-2\epsilon_1|^{-1}\).\par
    Depolarizing:
    \(Q=\frac{1+2\epsilon_1/3}{|1-4\epsilon_1/3|}\).
  }
& \tabnote{
    Amplification \(a=1+1.3/(\rho V\epsilon_0)\).\par
    Nonzero syndromes rejected.
  }
& \tabnote{Not considered.}
\tabularnewline
\noalign{\vskip 4pt}
\hline
\noalign{\vskip 4pt}

\methodcell{\textbf{\underline{SOLEM}}\\[-2pt]
  {\scriptsize Sec.~\ref{sec:model}}}
& \formula{
    e^{p_{1|L}^{-1}\epsilon_1^{-1}V\epsilon_L}
  }
& \parbox[t]{0.292\textwidth}{%
    \centering
    \(\displaystyle
      e^{\bar\lambda\rho V\epsilon_L}
      \left[
        1+\frac{V\epsilon_L}{p_{1|L}\epsilon_1}
        (\bar\lambda-1)^2
      \right]
    \)\par
    \vspace{2pt}
    {\scriptsize
      Reduction factors \(a_1,a_2\rightarrow
      a_c=a_{\rm EC}(\bar\lambda-1)\).
      \par
    }
  }
& \formula{
    \begin{aligned}[t]
      &e^{\bar\lambda\rho V\epsilon_L}
       \left[
         1+\frac{V\epsilon_L}{p_{1|L}\epsilon_1}
         (\bar\lambda-1)^2
       \right]
      \\[-2pt]
      &{}-\E{O}^2
       \left[
         e^{(2-\bar\lambda)\rho V\epsilon_L
            (2-\rho\epsilon_1)^2}F_0-1
       \right]
    \end{aligned}
  }
\tabularnewline
\noalign{\vskip 4pt}
\hline\hline
\end{tabular}
\endgroup
\end{table*}

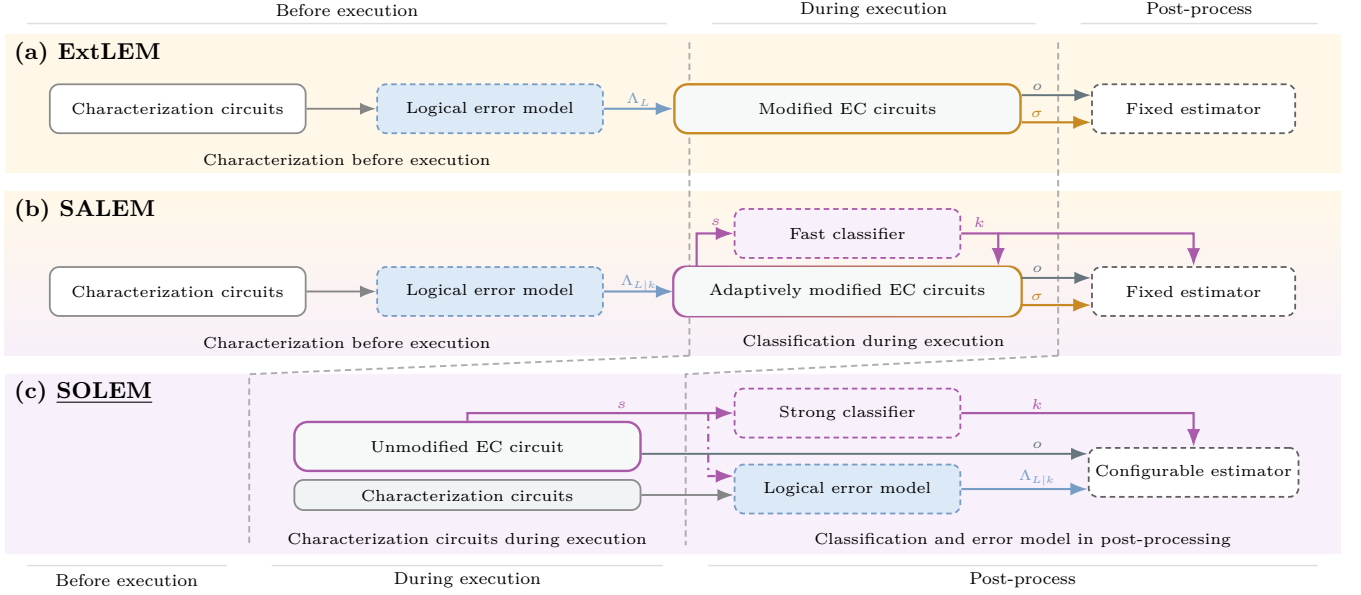
\begin{figure*}[t]
    \centering
    \resizebox{\textwidth}{!}{\begingroup
\providecommand{\figtext}{\rmfamily}
\begin{tikzpicture}

\definecolor{panelamber}{HTML}{FFF8E8}
\definecolor{panelpurple}{HTML}{F8F0FA}
\definecolor{samplingamber}{HTML}{C78922}
\definecolor{qpamber}{HTML}{F5D89C}
\definecolor{syndromepurple}{HTML}{A95AA9}
\definecolor{channelblue}{HTML}{DCE9F7}
\definecolor{channelstroke}{HTML}{789DC5}
\definecolor{neutralfill}{HTML}{F1F4F5}
\definecolor{executionfill}{HTML}{F7F9F9}
\definecolor{outcomegray}{HTML}{67767D}
\definecolor{ink}{HTML}{000000}
\definecolor{muted}{HTML}{000000}

\tikzset{
  x=1cm,y=.82cm,font=\figtext,line cap=round,line join=round,
  >={Latex[length=1.75mm,width=1.2mm]},
  panel/.style={sharp corners,draw=none},
  phase/.style={font=\figtext\mdseries\fontsize{6.5}{7.1}\selectfont,
                text=muted,inner sep=1pt},
  paneltitle/.style={font=\figtext\bfseries\fontsize{8.7}{9.4}\selectfont,
                    text=ink,anchor=west,inner sep=0pt},
  message/.style={font=\figtext\mdseries\fontsize{6.05}{6.6}\selectfont,
                  text=ink,align=center},
  bottommessage/.style={font=\figtext\fontsize{6.0}{6.6}\selectfont,
                        text=ink,align=center,inner sep=.4pt},
  timinglabel/.style={font=\figtext\fontsize{4.75}{5.2}\selectfont,
                      text=muted,align=center,inner sep=.3pt},
  quantum/.style={rounded corners=1.25mm,draw=ink!45,fill=white,line width=.68pt,
                  minimum height=6.5mm,align=center,inner xsep=2.7pt,inner ysep=1.7pt,
                  font=\figtext\mdseries\fontsize{6.05}{6.6}\selectfont,text=ink},
  execution/.style={rounded corners=1.8mm,draw=ink!48,fill=executionfill,
                    line width=.78pt,minimum height=6.5mm,minimum width=48mm,
                    align=center,font=\figtext\mdseries\fontsize{6.05}{6.6}\selectfont,
                    text=ink},
  algorithm/.style={rounded corners=1.2mm,dashed,draw=ink!58,fill=neutralfill,
                    dash pattern=on 2.1pt off 1.45pt,line width=.68pt,
                    minimum height=6.5mm,align=center,inner xsep=2.5pt,inner ysep=1.6pt,
                    font=\figtext\mdseries\fontsize{6.05}{6.6}\selectfont,text=ink},
  p2lc/.style={algorithm,draw=channelstroke,fill=channelblue},
  fast/.style={algorithm,draw=syndromepurple,fill=panelpurple},
  strong/.style={algorithm,draw=syndromepurple,fill=panelpurple},
  estimator/.style={algorithm,draw=ink!62,fill=white,minimum height=6.5mm,
                    font=\figtext\mdseries\fontsize{6.05}{6.6}\selectfont},
  data/.style={rounded corners=.85mm,fill=white,line width=.68pt,
               minimum height=4.2mm,inner xsep=1.35pt,inner ysep=.45pt,
               font=\fontsize{5.45}{5.95}\selectfont,text=ink},
  sdata/.style={data,draw=syndromepurple},
  xdata/.style={data,draw=samplingamber},
  odata/.style={data,draw=outcomegray},
  ldata/.style={data,draw=channelstroke},
  flow/.style={draw=ink!48,line width=.70pt,-{Latex}},
  sflow/.style={draw=syndromepurple,line width=.84pt,-{Latex}},
  charflow/.style={sflow,dash dot},
  xflow/.style={draw=samplingamber,line width=.84pt,-{Latex}},
  chflow/.style={draw=channelstroke,line width=.80pt,-{Latex}},
  oflow/.style={draw=outcomegray,line width=.74pt,-{Latex}},
  softlabel/.style={font=\figtext\mdseries\fontsize{5.45}{6.0}\selectfont,
                    text=muted,align=center},
  linelabel/.style={font=\fontsize{5.45}{5.95}\selectfont,
                    inner sep=.4pt,align=center}
}

\path[use as bounding box] (-.02,-0.4800) rectangle (17.77,9.1800);

\node[phase] at (4.55,8.9200) {Before execution};
\node[phase] at (11.55,8.9200) {During execution};
\node[phase] at (15.875,8.9200) {Post-process};
\draw[draw=ink!13,line width=.55pt] (.30,8.7000)--(8.80,8.7000);
\draw[draw=ink!13,line width=.55pt] (9.40,8.7000)--(13.70,8.7000);
\draw[draw=ink!13,line width=.55pt] (14.30,8.7000)--(17.45,8.7000);

\path[panel,fill=panelamber] (0,6.2800) rectangle (17.75,8.5200);
\path[panel,top color=panelamber,bottom color=panelpurple]
  (0,3.3200) rectangle (17.75,5.9800);
\path[panel,fill=panelpurple] (0,0.1200) rectangle (17.75,3.0200);
\foreach \x in {9.10,14.00}{
  \draw[draw=ink!32,dash pattern=on 2.4pt off 1.6pt,line width=.72pt]
    (\x,3.3200)--(\x,8.4200);
}
\draw[draw=ink!32,dash pattern=on 2.4pt off 1.6pt,line width=.72pt]
  (9.10,3.3200)--(3.25,3.0200)--(3.25,0.2500);
\draw[draw=ink!32,dash pattern=on 2.4pt off 1.6pt,line width=.72pt]
  (14.00,3.3200)--(9.05,3.0200)--(9.05,0.2500);

\node[bottommessage] (a-caption) at (4.55,6.5000)
  {Characterization before execution};

\node[quantum,minimum width=3.4cm,minimum height=6.5mm]
  (a-char) at (2.30,7.3200) {Characterization circuits};
\node[p2lc,minimum width=3.0cm,minimum height=6.5mm]
  (a-model) at (6.45,7.3200) {Logical error model};
\draw[flow] (a-char.east)--(a-model.west);

\node[execution,minimum width=46mm,draw=samplingamber,line width=.90pt]
  (a-run) at (11.20,7.3200) {Modified EC circuits};
\draw[chflow] (a-model.east)--(a-run.west);
\node[linelabel,text=channelstroke] at (8.425,7.4700) {$\Lambda_L$};

\node[estimator,minimum width=2.70cm,minimum height=6.5mm]
  (a-est) at (15.80,7.3200) {Fixed estimator};
\coordinate (a-title-bottom) at ($(a-char.north)+(a-char.south)-(a-caption.north)$);
\node[paneltitle,anchor=south west] at (.12,0 |- a-title-bottom) {(a) ExtLEM};
\draw[oflow] (a-run.east |- 11.75,7.5400)--(a-est.west |- 13.20,7.5400);
\node[linelabel,text=outcomegray] at (13.72,7.6800) {$o$};
\draw[xflow] (a-run.east |- 11.75,7.1000)--(a-est.west |- 13.05,7.1000);
\node[linelabel,text=samplingamber] at (13.72,7.2400) {$\sigma$};

\node[paneltitle] at (.12,5.6800) {(b) SALEM};
\node[bottommessage] at (4.55,3.5400)
  {Characterization before execution};
\node[bottommessage] at (11.55,3.5400) {Classification during execution};

\node[quantum,minimum width=3.4cm,minimum height=6.5mm]
  (b-char) at (2.30,4.3600) {Characterization circuits};
\node[p2lc,minimum width=3.0cm,minimum height=6.5mm]
  (b-model) at (6.45,4.3600) {Logical error model};
\draw[flow] (b-char.east)--(b-model.west);

\path[shade,left color=syndromepurple,right color=samplingamber,
      rounded corners=1.8mm] (8.87,3.9393) rectangle (13.53,4.7807);
\node[execution,minimum width=46mm,minimum height=6.5mm,draw=none,fill=executionfill]
  (b-run) at (11.20,4.3600)
  {Adaptively modified EC circuits};
\draw[chflow] (b-model.east)--(b-run.west);
\node[linelabel,text=channelstroke] at (8.425,4.5100) {$\Lambda_{L|k}$};

\node[fast,minimum width=3.0cm,minimum height=6.5mm] (b-fast) at (11.20,5.3000)
  {Fast classifier};
\node[estimator,minimum width=2.70cm,minimum height=6.5mm]
  (b-est) at (15.80,4.3600) {Fixed estimator};

\coordinate (b-selbow) at (9.196,5.3000);
\draw[sflow] (b-run.north -| b-selbow)--(b-selbow)--(b-fast.west);
\node[linelabel,text=syndromepurple] at (9.448,5.4800) {$s$};
\coordinate (b-ksplit) at (13.204,5.3000);
\draw[sflow] (b-fast.east)--(b-ksplit)--(b-run.north -| b-ksplit);
\node[linelabel,text=syndromepurple] at (12.952,5.4800) {$k$};
\draw[sflow] (b-ksplit)--(15.80,5.3000)--(b-est.north);

\draw[oflow] (b-run.east |- 11.70,4.5800)--(b-est.west |- 13.20,4.5800);
\node[linelabel,text=outcomegray] at (13.72,4.7200) {$o$};
\draw[xflow] (b-run.east |- 11.70,4.1400)--(b-est.west |- 13.05,4.1400);
\node[linelabel,text=samplingamber] at (13.72,4.2800) {$\sigma$};

\node[paneltitle] at (.12,2.7200) {(c) \underline{SOLEM}};
\node[bottommessage] at (6.15,0.3400)
  {Characterization circuits during execution};
\node[bottommessage] at (13.53,0.3400)
  {Classification and error model in post-processing};

\node[execution,minimum width=46mm,draw=syndromepurple,line width=.90pt]
  (c-run) at (6.15,1.8500) {Unmodified EC circuit};

\node[quantum,fill=neutralfill,minimum width=46mm,minimum height=4.0mm]
  (c-pchar) at (6.15,1.0634) {Characterization circuits};

\node[strong,minimum width=3.0cm,minimum height=6.5mm] (c-strong) at (11.20,2.3900)
  {Strong classifier};
\node[p2lc,minimum width=3.0cm,minimum height=6.5mm] (c-model) at (11.20,1.1600)
  {Logical error model};
\node[estimator,minimum width=2.70cm,minimum height=6.5mm] (c-est) at (15.80,1.4300)
  {Configurable estimator};

\coordinate (c-split) at (9.35,2.3900);
\draw[draw=syndromepurple,line width=.84pt]
  (c-run.north)--(6.15,2.3900)--(c-split);
\node[linelabel,text=syndromepurple] at (8.20,2.5300) {$s$};
\draw[sflow] (c-split)--(c-strong.west);
\draw[charflow] (c-split)--(9.35,1.3700)--(c-model.west |- 9.35,1.3700);
\draw[oflow] (c-run.east |- 8.60,1.7300)--(c-est.west |- 15.15,1.7300);
\node[linelabel,text=outcomegray] at (13.72,1.8700) {$o$};
\draw[flow] (c-pchar.east)--(c-model.west |- c-pchar.east);
\draw[sflow] (c-strong.east)--(15.80,2.3900)--(c-est.north);
\node[linelabel,text=syndromepurple] at (13.72,2.5300) {$k$};
\draw[chflow] (c-model.east)--(c-est.west |- c-model.east);
\node[linelabel,text=channelstroke] at (13.72,1.3500) {$\Lambda_{L|k}$};

\draw[draw=ink!13,line width=.55pt] (.30,-0.0800)--(2.95,-0.0800);
\draw[draw=ink!13,line width=.55pt] (3.55,-0.0800)--(8.75,-0.0800);
\draw[draw=ink!13,line width=.55pt] (9.35,-0.0800)--(17.45,-0.0800);
\node[phase] at (1.625,-0.3100) {Before execution};
\node[phase] at (6.15,-0.3100) {During execution};
\node[phase] at (13.53,-0.3100) {Post-process};

\end{tikzpicture}
\endgroup}%
    \caption{\textbf{Complete SOLEM workflow, and comparison to ExtLEM and SALEM.} Each panel is divided into a central step in which mitigation circuits are prepared and executed, with large boxes abstracting the detailed description at the top of Fig.~\ref{fig:overview_schematics}, as well as additional classical logic (dashed boxes) and quantum logic (full boxes) occurring before execution and in post-processing. The latter includes a final estimation step, abstracting the bottom of Fig.~\ref{fig:overview_schematics}. Stage widths are adjusted for readability. 
\textbf{(a)} ExtLEM obtains its characterization \textit{before} execution and uses it to prepare modified EC circuits, exposing the resulting estimator to biases due to changes in physical errors (drifts) in the time span between characterization and mitigation circuits. Observable outcomes \(\boldsymbol{o}\) and modification records \(\boldsymbol{\sigma}\) are used for estimation.
\textbf{(b)} SALEM likewise characterizes the relevant logical errors before execution. During execution, a `fast' classifier, based on soft-output from the decoder used for EC, maps each measured syndrome \(s\) to a class \(k\), which controls adaptive circuit modifications. Usage of a fast and therefore `weak' classifier limits the improvement in shot-overhead relative to ExtLEM. The classes \(\boldsymbol{k}\), logical outcomes \(\boldsymbol{o}\), and modification labels \(\boldsymbol{\sigma}\) are used for estimation. In both panels, logical error models are implicitly used in estimation (via e.g., usage of QP signs and normalization).
\textbf{(c)} SOLEM executes the unmodified EC circuit and retains its syndrome and observable records. Characterization circuits can be interleaved with execution, significantly reducing sensitivity to hardware drifts. Moreover,  dedicated characterization circuits may be redundant, as it may be possible to obtain a logical error model by processing the syndrome distribution \cite{wagner2023learninglogicalpaulinoise,xiao2026insitubenchmarkingfaulttolerantquantum,zheng2026efficientlearninglogicalnoise}. 
In both cases, the logical error model is obtained in post-processing. Classification is also performed in post-processing, allowing a `stronger' classifier, based on a more accurate but slower soft-output decoder, which cannot be used in real time. 
The resulting logical error model, class records and observable records determine the SOLEM estimator. Classification choices, as well as additional estimation choices,  can thus be configured freely, to optimize empirical error bars and additional performance metrics.  
    }
    \label{fig:workflows}
\end{figure*}

\textbf{Syndrome-based offline mitigation.} In this work we ask whether syndrome data can not only serve as \textit{a} resource to reduce the LEM shot overhead, but as \textit{the} resource on which LEM is based. Our basic insight is that syndrome-conditioned expectation values correspond to the ideal logical circuit in the presence of \textit{distinct} logical error channels, which may therefore replace the circuit modifications used in standard EM, and inherited by ExtLEM and SALEM. This raises the central technical question addressed in this work:

\textit{When, and how efficiently, can empirical syndrome-conditioned expectation values of a given faulty EC circuit be used to estimate ideal logical expectation values?}

We address this question constructively, by introducing \emph{syndrome-based offline} LEM (SOLEM): a framework that 
estimates ideal logical expectation values by post-processing 
syndrome records and logical measurement outcomes, together with a characterization of syndrome-conditioned logical errors.
SOLEM entails no added real-time classical processing or adaptivity, and, moreover, executes \textit{only} the given EC circuit, without \emph{any} modification. In this sense, SOLEM is simpler to implement than both SALEM and ExtLEM, and may be viewed as an inherent feature of EC circuits.

Within the SOLEM framework, we develop three specific constructions, offering different performance tradeoffs. First, Fig.~\ref{fig:overview_schematics}(c)  illustrates how quasi-probability (QP) representations using the naturally occurring logical channels can (in principle) reproduce the identity channel, yielding an offline logical version of PEC. Exact cancellation, leading to unbiased estimators, is possible when the identity lies in the linear span of syndrome-conditioned channels, while the shot overhead depends on both these channels and the distribution of syndromes. As in standard PEC, approximate representations introduce a controlled bias-overhead trade-off, allowing a lower shot overhead in exchange for residual errors that bound the estimation bias. 
Second, we use QP representations to \textit{partially} reduce logical errors (and estimate derivatives with respect to logical channels), leading to a heuristic, but often more efficient, offline logical version of \textit{zero-noise extrapolation} (ZNE) \cite{Temme2017QEM,Endo2018QEM}. 
Third, we develop \textit{global regression} (GR), a heuristic construction that assigns each syndrome record a parametric
``global decay'' factor, and fits the syndrome-conditioned means of observed outcomes to the products of these factors and the unknown ideal expectation value.

\textbf{Results.}
Our constructions use a characterization of syndrome-conditioned logical
error channels associated with individual logical layers. We obtain this characterization using the
physical-to-logical characterization (P2LC) framework \cite{aharonov2025syndromeawaremitigationlogical}, which maps a
circuit-level physical noise model to effective logical channels while
accounting for temporal correlations between layers.
We review the required formalism in Sec.~\ref{sec:setup}, 
and give a general description of the SOLEM framework in Sec.~\ref{sec:solem_family}.
We then present the 
QP-based SOLEM constructions in Sec.~\ref{sec:solem-qp-based} and
introduce SOLEM-GR in Sec.~\ref{sec:solem-gr}.

We first assess the bias and sampling overhead of SOLEM-PEC, SOLEM-ZNE,
and SOLEM-GR in circuit-level simulations of the Steane-code logical
memory shown in Fig.~\ref{fig:memory-simulation}(a); see also
Sec.~\ref{sec:circ_leve_simul}. We use P2LC to map the circuit-level
noise model to local syndrome-conditioned logical error channels (Fig.~\ref{fig:memory-simulation}(b) and (c)).
The characterized channels, syndrome records, and logical measurement
outcomes are then used to construct the three SOLEM estimators entirely
in post-processing. As shown in Fig.~\ref{fig:memory-simulation}(d), we
resolve no statistically significant bias for any of the three
estimators across the simulated circuit depths. SOLEM-PEC achieves a
lower sampling overhead than external PEC, while SOLEM-ZNE and SOLEM-GR
further reduce the overhead in this setting, illustrating the potential
gains enabled by their additional modeling assumptions.
Fig.~\ref{fig:solem_pec_panels}(b) shows the effect of finer syndrome partitions, and the importance of taking input errors into account.

To obtain qualitative insight into these overheads and compare SOLEM with existing LEM methods, we introduce an analytically solvable effective logical-noise model in Sec.~\ref{sec:model}.
The model allows an arbitrary circuit and
observable, but assumes independent binary syndrome labels and
site-independent, ``parallel'' syndrome-conditioned logical channels \(\Lambda_{L|k} = (1-\epsilon_k)I + \epsilon_k E\) 
with \(k\in \{0,1\}\), and error map $E$.

Within this model, the expectation conditioned on \(l\) nontrivial
syndromes decomposes as
\(\mathbb E[o|N_1=l]=\sum_r K_{r,l}C_r\), separating the
error-count statistics \(K_{r,l}\) from the ideal circuit's response
\(C_r\) to \(r\) inserted errors
(Prop.~\ref{prop:ev-cond-N1}). 
When the latter is exponentially decaying, the heuristic models used by SOLEM-ZNE
and SOLEM-GR are exact
(Prop.~\ref{prop:exp_Cr}).
SOLEM-PEC is unbiased within the parallel-channel model,
while SOLEM-ZNE and SOLEM-GR are asymptotically unbiased
under this additional condition.
Their leading-order sampling overheads are summarized
in Table~\ref{tab:overhead-summary}.

Figures~\ref{fig:pec-lam} and \ref{fig:solem-zne-overheads} illustrate
the resulting overhead expressions for representative parameter
choices. In the regimes shown, SOLEM-PEC has an overhead comparable to,
and in some cases lower than, SALEM-PEC. 
When the observable significantly decays due to noise, SOLEM-GR and SOLEM-ZNE have nearly equal
overheads and outperform the other methods included in the comparison.

For the conditional Bernoulli model associated with the SOLEM-GR
response, we additionally derive the Fisher-information lower bound for asymptotically unbiased estimators based solely on the
coarsened observations \(\{(|\bsk_i|,o_i)\}_{i=1}^{N}\), where \(|\bsk_i|\) is the number of 1 syndromes in shot \(i\) (Prop.~\ref{prop:gr-fisher}). 
Under standard assumptions, the maximum-likelihood estimator
asymptotically attains the exact bound. We further show that unweighted
least squares agrees with this bound through quadratic order in the
conditional mean.

Finally, we investigate the bias arising from the heuristic response
models used by SOLEM-GR and SOLEM-ZNE in effective logical-channel
simulations of non-Clifford Floquet circuits
(Sec.~\ref{Sec: effective logical simulation}). Across the circuits
considered, fine-grained SOLEM-GR has a lower sampling overhead than
SOLEM-ZNE, while the 
deviations of the two
estimators from the ideal value are strongly correlated
correlated (Fig.~\ref{fig:flouqet_sim}). Consequently, these simulations do not resolve a bias
specific to the SOLEM-GR heuristic compared to the commonly used ZNE heuristic. 

To examine that heuristic more directly, we also construct a binary
parallel-channel partition and compute the circuit-response coefficients
\(C_r\). For most of the tested circuits, \(C_r\) is accurately
approximated by the exponential form assumed in
Prop.~\ref{prop:exp_Cr}. 
For the selected circuit depths shown in
Fig.~\ref{fig:cr_in_floquet}, statistically significant deviations of
the SOLEM-GR and SOLEM-ZNE estimates from the ideal value coincide with systematic
departures of \(C_r\) from an exponential response. 
This provides a
direct numerical test of the principal modeling assumption underlying
SOLEM-GR and of ZNE and identifies its breakdown as a source of estimator bias.

\textbf{Advantages of offline mitigation.} We discuss here three advantages due to SOLEM's offline nature, related to error characterization, syndrome classification, and estimator configuration.

A first advantage of SOLEM over ExtLEM and SALEM is related to the characterization of logical errors, which is needed, in some form, for most proposed LEM methods \cite{Suzuki2022, Piveteau2021, Lostaglio2021, Xiong2020, Tsubouchi2024, aharonov2025syndromeawaremitigationlogical, Zhou2025ErrorMitigation,dincă2025errormitigationlogicalcircuits, tsubouchi2026quantumadvantagessyndromeawarenoisy,yuan2026,kumar2026}. This characterization generally involves running \textit{characterization circuits} which are sensitive to hardware errors, and classically processing the resulting measurement outcomes to obtain a logical error model. In ExtLEM and SALEM, this model is required to determine the corresponding \textit{mitigation circuits}, and must therefore be obtained \textit{before} these are executed, see Fig.~\ref{fig:workflows}(a)-(b). This workflow is therefore sensitive to \textit{drifts} - temporal changes in hardware that often make learned error models inaccurate by the time mitigation circuits are executed, requiring careful treatment \cite{Kim2023Utility,aharonov2025reliablehighaccuracyerrormitigation,Leviatan2026PrethermalFloquet}. In contrast, SOLEM requires an error model only in post-processing, allowing characterization circuits to be temporally-interleaved with mitigation circuits, significantly reducing sensitivity to drifts, see Fig.~\ref{fig:workflows}(c).

A particularly appealing workflow can arise by combining SOLEM with syndrome-based characterization \cite{wagner2023learninglogicalpaulinoise,xiao2026insitubenchmarkingfaulttolerantquantum,zheng2026efficientlearninglogicalnoise}. These  works found that unmodified EC circuits can serve as characterization circuits for both physical and logical error models, obtained by processing the naturally occurring distribution of syndromes. Thus, a complete logical characterization \textit{and} mitigation workflow is in principle built into EC circuits, requiring only classical post-processing of syndromes and logical measurement outcomes. 

Another advantage of SOLEM's offline nature is related to the \textit{classification} of syndromes. As mentioned above, SALEM's classical processing can be eased by coarse-graining syndromes into a few carefully constructed classes, reducing much of the real-time task to a mapping of syndromes to classes \cite{aharonov2025syndromeawaremitigationlogical}, see Fig.~\ref{fig:workflows}(b). Useful and efficient classifiers can be obtained by thresholding on \textit{soft-output} from the decoder already used for EC \cite{aharonov2025syndromeawaremitigationlogical, Zhou2025ErrorMitigation,dincă2025errormitigationlogicalcircuits}, which is a measure of the decoder's confidence in its chosen recovery operation \cite{gidney2023yokedsurfacecodes,
smith2024mitigating,
meister2024efficientsoftoutputdecoderssurface,
kishi2026softoutput,
aydin2026softdecoding,dentelski2026neuralconfidence}. 

Within SOLEM, classification can still be useful, to reduce all (subsequent) post-processing from the large set of syndromes to a small number of classes. Since this classification isn't needed during execution, it need not be based on the `fast but weak' decoder used for EC, which will typically have sub-optimal accuracy due to the stringent runtime constraints of real-time processing. Instead, it is natural to deploy a `strong but slow' decoder in post-processing (Fig.~\ref{fig:workflows}(c)), and classify syndromes based on its more informative soft-output, leading to a reduced shot overhead. 

A concrete example of such a two-decoder workflow was studied in
Ref.~\cite{aharonov2025syndromeawaremitigationlogical}, in the special
case of binary SALEM with a rejected subset, which can be implemented
non-adaptively~\cite{aharonov2025syndromeawaremitigationlogical,
Zhou2025ErrorMitigation}. There, a fast perfect-matching decoder \cite{higgott2021pymatching} was used for EC, while a stronger
tensor-network decoder \cite{piveteau2023tensornetworkdecoding}
was used for classification in post-processing. Within SOLEM, such a workflow is \textit{generally} applicable, without restricting to binary classifications or sub-optimal mitigation choices. %

A third advantage of SOLEM's offline nature is related to the choice of estimator. Within ExtLEM and SALEM, the choice of estimator affects the choice of mitigation circuits, and is therefore commonly fixed in pre-processing; changing it in post-processing
generally incurs an additional sampling overhead. In contrast, multiple SOLEM estimators can be freely generated from the same measured data, allowing for an optimization of the estimator over empirical error bars, as well as additional performance metrics. As a concrete example, when performing external ZNE, there is an optimal amplification factor which minimizes the resulting shot overhead, but it depends on an a priori unknown ratio between ideal and noisy expectation values \cite{filippov2024scalability, Haghshenas2026DigitalMagnetism,Leviatan2026PrethermalFloquet}. Within SOLEM-ZNE, we find a related optimal amplification (Sec.~\ref{sec:model-overheads}). It too depends on a priori unknown data, but within SOLEM-ZNE this data is only needed a posteriori, where it can be estimated. 

 \textbf{Related work.} 
The closest literature to SOLEM-PEC (and SOLEM-ZNE) can be found in Ref.~\cite{aharonov2025reliablehighaccuracyerrormitigation}, where `retroactive resampling' is presented as a means to overcome drifts in physical QP-based EM. In this method, an initial noise model is used to generate a QP distribution from which mitigation circuits are sampled. An updated noise model is then used to `resample' the resulting data in post-processing, such that it corresponds to an updated QP distribution. This eliminates biases due to drifts, at the cost of an increased shot overhead. Note, however, that (i) in SOLEM-PEC there is no need for an initial noise model and a corresponding ensemble of modified circuits, and (ii) SOLEM-PEC can significantly improve the shot overhead relative to ExtPEC, rather than always increasing it.

Reference~\cite{dincă2025errormitigationlogicalcircuits} proposed a maximum-likelihood method which is close to SOLEM-GR, sharing the underlying principle of mapping syndrome-related data to expectation-value decay constants. However, Ref.~\cite{dincă2025errormitigationlogicalcircuits} relies on the ability to exactly map local logical errors to their effect on the observable in memory circuits, and does not explain how to generalize the approach to general circuits where such a mapping may be computationally intractable. In contrast, SOLEM-GR is based on a heuristic model that bypasses this computational barrier. Though no single model is valid for all circuits, the one we focused on performs well in the non-Clifford circuits studied in Sec.~\ref{Sec: effective logical simulation}. Even in the deep circuit-level memory simulations of Fig.~\ref{fig:memory-simulation}, we do not rely on the ability to exactly map errors to decays; instead we utilize our heuristic model, which does not hold exactly but still performs very well. 

Additionally,  Ref.~\cite{dincă2025errormitigationlogicalcircuits} relies on soft-output from the `fast but weak' decoder used for EC, for both classification and characterization. However, since both procedures are only required
in post-processing, we advocate using stronger
soft-output decoders (such as tensor-network decoders)
for classification and P2LC for more accurate
characterization, as shown in Fig.~\ref{fig:workflows}.

In this paper we constructively studied the opportunities and limitations of LEM utilizing syndrome data entirely in post-processing. A complementary approach was taken in  Ref.~\cite{tsubouchi2026quantumadvantagessyndromeawarenoisy}, which compared the classical and quantum Fisher information to identify a separation between the optimal shot overhead of offline and adaptive syndrome-based LEM. However, we find it hard to directly relate the circuit-level setup considered here with the simplified setting of Ref.~\cite{tsubouchi2026quantumadvantagessyndromeawarenoisy}, which considered a single round of EC with code-capacity errors, followed by ideal logical measurements. As an example, the optimal `offline' method found in Ref.~\cite{tsubouchi2026quantumadvantagessyndromeawarenoisy} corresponds to the `fine-grained' SALEM-PEC of Ref.~\cite{aharonov2025syndromeawaremitigationlogical}, which in general logical circuits is highly adaptive. It would be interesting to extend the information-theoretic approach to the circuit-level setting, and understand whether the methods proposed here make optimal use of syndrome data in post-process. In this work we only addressed this optimality question in the context of SOLEM-GR, under its assumed model.

\textbf{Future directions.} The ability to mitigate logical errors using the syndrome data of unmodified EC circuits opens several directions beyond the constructions developed here.

A central question is how to use syndrome information optimally for mitigation: what are the limits of purely offline estimation, and how much further can adaptive quantum control improve performance? SOLEM's advantage over some SALEM constructions suggests that adaptive protocols could benefit from more effective syndrome-resolved post-processing. Conversely, in our comparisons, external ZNE outperforms the studied syndrome-based methods when the observable is very weakly sensitive to noise, a regime motivated by the error-dilution phenomenon studied in Ref.~\cite{Granet_2025}. These comparisons motivate a better understanding of how much improvement syndrome information can provide, both with and without adaptive control. As discussed above, extending the information-theoretic approach of Ref.~\cite{tsubouchi2026quantumadvantagessyndromeawarenoisy} to more general fault-tolerant computations could help determine how close existing methods come to the best achievable performance, and when adaptivity provides an additional advantage.

Our analysis establishes the accuracy and sampling advantages of SOLEM-ZNE and SOLEM-GR in the solvable model, while circuit-level memory and effective logical non-Clifford simulations support their applicability beyond that setting. A natural next step is to characterize their bias--variance trade-offs for broader circuit, code, and noise families. Beyond the usual extrapolation assumption, SOLEM-GR models the response using a
a parametric global decay factor
of the syndrome record, motivating a study of how this additional assumption affects extrapolation accuracy beyond the settings considered here. For both methods, can departures from the assumed exponential response be identified in the measured data and related quantitatively to extrapolation bias?

For SOLEM-ZNE, natural directions include multi-point extrapolation and derivative-based estimators constructed from the same experimental records. The number and placement of noise levels, the derivative orders, and the extrapolation model can be optimized jointly, accounting for the covariance between the resulting estimates. It remains to determine when these constructions improve the bias--variance trade-off beyond the two-point method, and its point-and-derivatives limit, analyzed here.

The QP constructions demonstrated here motivate studying their scalability to codes encoding many logical qubits per block, including quantum LDPC codes~\cite{bravyi2024high}. In these settings, the linear constraints defining SOLEM-PEC and SOLEM-ZNE involve a larger logical channel space. Under what conditions do they admit accurate solutions with low sampling overhead, and can such solutions be found efficiently by exploiting the code structure and an appropriate syndrome classification? These questions also motivate the co-design of error correction and mitigation~\cite{kumar2026}. Can a code or decoder with a higher logical error rate nevertheless enable more efficient observable estimation by providing more informative syndrome data?

A particularly appealing prospect is to obtain both characterization and mitigation from the unmodified error-corrected computation itself: under what conditions are its syndrome and measurement records sufficient for both tasks, and when can they be performed efficiently? Our circuit-level demonstration starts from a given physical noise model and uses physical-to-logical characterization to construct SOLEM estimators. A natural next step is to determine when the required characterization can itself be learned efficiently from the computation's syndrome and measurement records, building on syndrome-based noise-learning methods~\cite{wagner2023learninglogicalpaulinoise,xiao2026insitubenchmarkingfaulttolerantquantum,zheng2026efficientlearninglogicalnoise}. Equally important is the precision required: how do errors in syndrome-conditioned channels and probabilities propagate into the mitigated estimate, and how does this sensitivity depend on the code, syndrome classification, and estimator? Answering these questions would establish the characterization cost required to realize the predicted gains, including the effects of model approximations and truncated syndrome histories.

Finally, SOLEM's compatibility with unmodified error-corrected circuits offers a direct route to experimental demonstrations of logical error mitigation. An important challenge is to realize its predicted gains in a complete workflow that combines characterization, syndrome processing, and mitigation, while accommodating non-Markovian noise, hardware drift, and characterization inaccuracies at a tractable classical cost.

\vspace{10pt}

\textbf{Acknowledgments.} We thank  our colleagues at Qedma for helpful discussions and  feedback during this project. 

\vspace{10pt}

\textbf{Competing interests.}  This work describes methods that are the subject of patent applications filed by Qedma Quantum Computing Ltd. ~\cite{qedma2025patent}.

\section{Preliminaries}\label{sec:setup}

\subsection{Error-corrected circuits and syndrome conditioning}

We follow the formalism of
Ref.~\cite[Appendices 1--2]{aharonov2025syndromeawaremitigationlogical}.
Consider a faulty error-corrected circuit $C=G_D\cdots G_1$.
Each layer $G_j$ includes gadgets implementing logical gates, syndrome measurement, decoding and recovery. Every physical operation in $G_j$ carries a physical error (or `fault') channel. To describe syndrome-conditioned quantities, we write $G_j$ as a `quantum instrument',
\begin{equation}
    G_j=\sum_{s_j\in S^{(j)}}G_{j,s_j},
\end{equation}
where the completely positive (CP) and trace-nonincreasing map $G_{j,s_j}$ corresponds to syndrome outcome $s_j$.
The sum $G_j$ is trace-preserving (TP), and we denote its intended logical operation by $g_j$.

We use the `super-bra-ket' notation,
with $\kket{\rho}$ denoting a density operator
and $\langle\!\langle O |\rho\rangle\!\rangle=\operatorname{Tr}(O\rho)$.
Let $\kket{G_\psi}$ be the state produced by
faulty error-corrected preparation
of the intended initial logical state $\kket{\psi}$.
Similarly, $\bbra{G_O}$ is the linear functional
giving the mean outcome of the faulty error-corrected
measurement of a logical observable $O$.
The corresponding ideal expectation value is
\begin{equation}
    \langle O\rangle
    =\bbra{O}g_D\cdots g_1\kket{\psi}.
\end{equation}
We assume $\|O\|_\infty\leq1$,
so measurement outcomes satisfy $o\in[-1,1]$.
For $O^2=I$, $o\in\{-1,1\}$.

Each circuit execution (`shot') produces a syndrome record
$\boldsymbol{s}=(s_1,\ldots,s_D)$ and a logical outcome $o$.
The joint syndrome distribution and conditional mean are
\begin{align}
    p_{\boldsymbol{s}}
    &=\bbra{I}G_{D,s_D}\cdots G_{1,s_1}\kket{G_\psi},
    \label{eq:ps}\\
    \mathbb{E}[o|\boldsymbol{s}]
    &=\bbra{G_O}G_{D,s_D}\cdots G_{1,s_1}\kket{G_\psi}/
    p_{\boldsymbol{s}},
    \label{eq:physical-conditional}
\end{align}
where conditioning is restricted to records of nonzero probability.
Our goal is to obtain a purely logical description
of these quantities, in terms of logical error channels associated with \textit{individual} layers. This is non-trivial because logical errors are generally due to faults occurring in multiple layers. 

\subsection{Exact predictive logical error channels}
\label{subsec:logical error ch}

The predictive logical error channels $\Lambda^{(j)}_{L\mid s_{1:j}}$ constructed in
Ref.~\cite[Appendix 2]{aharonov2025syndromeawaremitigationlogical}
reproduce the faulty expectation value, 
\begin{equation}
    \label{eq:predictive-s}
    \mathbb{E}[o|\boldsymbol{s}]
    =\bbra{O}\Lambda^{(D)}_{L\mid s_{1:D}}g_D
    \cdots\Lambda^{(1)}_{L\mid s_1}g_1\kket{\psi},
\end{equation}
where $s_{1:j}=(s_1,\ldots,s_j)$.
The corresponding conditional probabilities $p^{(j)}_{s_j\mid s_{1:j-1}}$ reproduce the faulty syndrome distribution, 
\begin{align}
    p_{\boldsymbol{s}}
    &=\prod_{j=1}^{D}p^{(j)}_{s_j\mid s_{1:j-1}}
    \label{eq:chain-s}.
\end{align}

The iterative construction of these objects is depicted in Fig.~\ref{fig:solem-iterative-construction}, for the simpler syndrome-averaged case. In the syndrome-conditioned case, a given input-error channel $\Lambda^{(j)}_{\mathrm{in}\mid s_{1:j-1}}$ is mapped by $G_j$ to an output error channel $\Lambda^{(j)}_{\mathrm{out}\mid s_{1:j}}$ (and a corresponding conditioned syndrome distribution $p^{(j)}_{s_j\mid s_{1:j-1}}$). The output is then factored as the product of a logical error channel and an input to the next layer,
\begin{equation}
    \Lambda^{(j)}_{\mathrm{out}\mid s_{1:j}}
    =\Lambda^{(j+1)}_{\mathrm{in}\mid s_{1:j}}
    \Lambda^{(j)}_{L\mid s_{1:j}}.\label{Eq: out=in*L}
\end{equation}
The logical channel is identified by acting on the output with the \textit{fault-free} version of the next layer, isolating the piece which is not corrected. These ideal layers are reference operations;
no component in the executed circuit is assumed to be fault-free. Equation~\eqref{Eq: out=in*L} then defines the next input, which is \textit{correctable}, in the sense that it will be corrected by the next layer, if no faults occur within it. Thus, each logical channel $\Lambda^{(j)}_{L\mid s_{1:j}}$ captures logical errors that appear at the output of $G_j$, but not earlier. Preparation and measurement errors are included
through the first input-error map,
$\Lambda_{\mathrm{in}}^{(1)}\kket{\psi}=\kket{G_\psi}$,
and the final logical error map, respectively.

Note that the definition of the input maps can make them non-CP, which passes over to the resulting logical error maps. We nevertheless abuse terminology and refer to these as `logical error channels'. Their important features are that they are predictive, and therefore serve as targets for mitigation; in the case of SOLEM, they are also used as building blocks.
Additionally, these exact channels serve as starting points for local approximations, as described in Sec.~\ref{Sec: P2LC}. 

\begin{figure}[t]
    \centering
    \resizebox{\columnwidth}{!}{%
        \begin{tikzpicture}[
  font=\small,
  >=Stealth,
  block/.style={draw, rounded corners=3pt, thick, minimum height=16mm, align=center},
  big/.style={block, minimum width=10mm, fill=pink!60},
  logical/.style={block, minimum width=10mm, fill=white},
  inchan/.style={
    block,
    minimum width=10mm,
    shade,
    left color=white,
    right color=pink!60,
    shading=axis
  },
  braceLabel/.style={midway, yshift=10pt}
]

\node[big] (G1) {$G_1$};
\node[big, right=1mm of G1] (G2top) {$G_2$};
\node[big, right=1mm of G2top] (G3top) {$G_3$};
\node[right=1mm of G3top] (dotstop) {$\cdots$};

\draw[thick] ([yshift=2pt,xshift=2pt]dotstop.east) -- ([yshift=-2pt,xshift=2pt]dotstop.east);
\draw[->, thick] ([xshift=2pt] dotstop.east) -- ++(5mm,0);

\node[logical, yshift=-2.5cm] (g1a) {$g_1$};
\node[logical, right=1mm of g1a] (lambda1a) {$\Lambda_L^{(1)}$};
\node[inchan, right=1mm of lambda1a] (lambda_in_2) {$\Lambda_{\mathrm{in}}^{(2)}$};
\node[big, right=1mm of lambda_in_2] (G2mid) {$G_2$};
\node[big, right=1mm of G2mid] (G3mid) {$G_3$};
\node[right=1mm of G3mid] (dotsmid) {$\cdots$};

\draw[decorate, line width=1.3pt, decoration={brace, amplitude=5pt}]
  (lambda1a.north west) -- (lambda_in_2.north east)
  node[braceLabel, yshift=5pt] {$\Lambda_{\mathrm{out}}^{(1)}$};

\draw[thick] ([yshift=2pt,xshift=2pt]dotsmid.east) -- ([yshift=-2pt,xshift=2pt]dotsmid.east);
\draw[->, thick] ([xshift=2pt] dotsmid.east) -- ++(5mm,0);

\node[logical, yshift=-5cm] (g1b) {$g_1$};
\node[logical, right=1mm of g1b] (lambda1b) {$\Lambda_L^{(1)}$};
\node[logical, right=1mm of lambda1b] (g2b) {$g_2$};
\node[logical, right=1mm of g2b] (lambda2b) {$\Lambda_L^{(2)}$};
\node[inchan, right=1mm of lambda2b] (lambda_in_3) {$\Lambda_{\mathrm{in}}^{(3)}$};
\node[big, right=1mm of lambda_in_3] (G3bot) {$G_3$};
\node[right=1mm of G3bot] (dotsbot) {$\cdots$};

\draw[decorate, line width=1.3pt, decoration={brace, amplitude=5pt}]
  (lambda2b.north west) -- (lambda_in_3.north east)
  node[braceLabel, yshift=5pt] {$\Lambda_{\mathrm{out}}^{(2)}$};

\end{tikzpicture}%
    }
    \caption{%
        Iterative construction of predictive logical error channels; syndrome-averaged case.
        Faulty EC layers (pink) are successively replaced
        by ideal logical operations and logical error maps (white).
        The gradient input-error map connects the logical prefix
        to the remaining physical circuit;
        braces identify the output-error factorization
        $\Lambda_{\mathrm{out}}^{(j)}
        =\Lambda_{\mathrm{in}}^{(j+1)}\Lambda_L^{(j)}$. Reproduced from 
        Ref.~\cite[Fig.~4]{aharonov2025syndromeawaremitigationlogical}.%
    }
    \label{fig:solem-iterative-construction}
\end{figure}
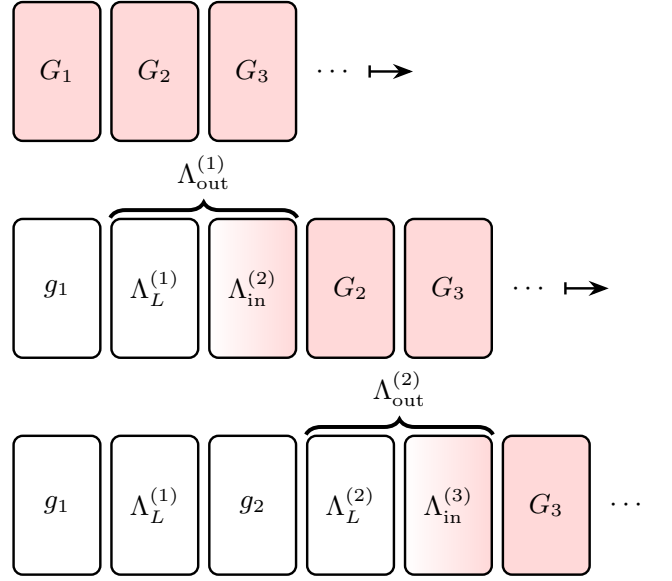

\subsection{Syndrome classification}
To simplify the post-processing of syndromes, we may classify each syndrome
into one of $K$ disjoint subsets,
$S^{(j)}=\bigcup_{k=0}^{K-1}S^{(j)}_k$, using a classifier  $\kappa_j:s_j\mapsto k_j$. As discussed in Sec.~\ref{Sec: intro}, useful and efficient classifiers can be obtained using soft-output decoders  \cite{aharonov2025syndromeawaremitigationlogical, Zhou2025ErrorMitigation,dincă2025errormitigationlogicalcircuits}. The corresponding coarse-grained instrument is
\begin{equation}
    G_{j,k_j}:=\sum_{s_j\in S^{(j)}_{k_j}}G_{j,s_j},
    \qquad G_j=\sum_{k_j=0}^{K-1}G_{j,k_j}.
\end{equation}
For a classified record $\boldsymbol{k}=(k_1,\ldots,k_D)\in [K]^D$,
the same construction gives
\begin{align}
    p_{\boldsymbol{k}}
    &:=\bbra{I}G_{D,k_D}\cdots G_{1,k_1}\kket{G_\psi}
    \nonumber\\
    &=\prod_{j=1}^{D}p^{(j)}_{k_j\mid k_{1:j-1}},
    \label{eq:pk}\\
    \mathbb{E}[o|\boldsymbol{k}]
    &:=\bbra{G_O}G_{D,k_D}\cdots G_{1,k_1}\kket{G_\psi}/p_{\boldsymbol{k}}
    \nonumber\\
    &=\bbra{O}\Lambda^{(D)}_{L\mid k_{1:D}}g_D
    \cdots\Lambda^{(1)}_{L\mid k_1}g_1\kket{\psi},
    \label{eq:predictive-k}
\end{align}
using the coarse-grained instruments.
For brevity, we sometimes write
$\Lambda^{(j)}_{L\mid\boldsymbol{k}}:=\Lambda^{(j)}_{L\mid k_{1:j}}$
and collect them into
$\boldsymbol{\Lambda}_{L\mid\boldsymbol{k}}
=(\Lambda^{(j)}_{L\mid\boldsymbol{k}})_{j=1}^{D}$. 

The data from each circuit execution (`shot')
is a pair $(\boldsymbol{k},o)$. Equation~\eqref{eq:predictive-k} provides a purely logical description for the expectation of $o$ given $\boldsymbol{k}$.

\subsection{Physical-to-logical characterization \label{Sec: P2LC}}

Physical-to-logical characterization (P2LC)
maps a physical noise model and a description of the EC circuit
to \textit{approximate} logical error maps which are computed locally, and in parallel, as opposed to the one global serial computation of the exact  construction \cite[Appendix 2.c]{aharonov2025syndromeawaremitigationlogical}. Given a `recent history' parameter $\mathsf{h}$, an approximate channel for the $j$th layer,
\begin{equation}
    \widetilde{\Lambda}^{(j)}_{L\mid k_{j-\mathsf{h}:j}}\approx \Lambda^{(j)}_{L\mid k_{1:j}},
\end{equation}
is obtained by replacing the input channel $\Lambda^{(j-\mathsf{h})}_{\mathrm{in}\mid k_{1:j-\mathsf{h}-1}}$ of the exact construction by the identity channel; see Fig.~\ref{fig:memory-simulation}(b) for a depiction of the case $\mathsf{h}=1$. This amounts to `forgetting' the far history beyond $\mathsf{h}$ layers backwards, including the correctable errors at that time and the syndrome classes $k_{1:j-\mathsf{h}-1}$ they generally depend on. This also leads to approximate syndrome probabilities,  $\widetilde{p}^{(j)}_{k_j|k_{j-\mathsf{h}:j-1}}\approx p^{(j)}_{k_j\mid k_{1:j-1}}$.

In a standard  $t$-fault-tolerant setting, $\mathsf{h}=1$ suffices to capture the leading-order syndrome-averaged logical errors $\sim\epsilon^{t+1}$, where $\epsilon$ is the physical error rate \cite{aharonov2025syndromeawaremitigationlogical}. A similar guarantee holds when conditioning on likely syndrome-classes, with an $O(1)$ probability. We validate our approximate channels for rare classes empirically, by comparing approximations obtained with different history lengths, and making sure they predict faulty expectation values accurately.  

Since the history parameter $\mathsf{h}=\mathsf{h}_j$ generally depends on the layer $j$ (at least in the beginning of the circuit), we find it convenient to define $h_j:=k_{j-\mathsf{h}_j:j-1}$, and write the approximate quantities as $\widetilde{\Lambda}^{(j)}_{L|k_j,h_j},\widetilde{p}^{(j)}_{k_j|h_j}$. We denote all approximate local channels and probabilities by 
\begin{align}
    \widetilde{\boldsymbol{\Lambda}}_{L|\boldsymbol{k}}=(\widetilde{\Lambda}^{(j)}_{L\mid k_j,h_j})_{j=1}^D,\ \ \ \ \widetilde{\boldsymbol{p}}_{\boldsymbol{k}}=(\widetilde{p}^{(j)}_{k_j|h_j})_{j=1}^D.    
\end{align}

\section{SOLEM framework}
\label{sec:solem_family}

SOLEM estimates the ideal logical expectation value
$\langle O\rangle$ by post-processing measurement records
$\{(\boldsymbol{k}_i,o_i)\}_{i=1}^N$ from $N$ repeated executions of a single, unmodified
error-corrected circuit. These records are generally supplemented by an approximate characterization
of the syndrome-conditioned logical errors
and syndrome probabilities,
\begin{equation}
    \widetilde{\boldsymbol{\Lambda}}_{L|}
    :=
    \bigl(
        \widetilde{\boldsymbol{\Lambda}}_{L\mid\boldsymbol{k}},\widetilde{\boldsymbol{p}}_{\boldsymbol{k}}
    \bigr)_{\boldsymbol{k}}.
\end{equation}
Note that we do not compute logical error channels for every global record $\boldsymbol{k}$; each layer $j$ has a much smaller number of classes $k_j$ and possible recent histories $h_j$. 

With this notation, a SOLEM estimator has the general form
\begin{equation}
    \mathsf{Est}
    =
    \Phi\!\left(
        \{(\boldsymbol{k}_i,o_i)\}_{i=1}^{N};
        \widetilde{\boldsymbol{\Lambda}}_{L|}
    \right),
\end{equation}
where $\Phi$ is a classical function 
aimed at estimating $\langle O\rangle$. Specific constructions may use only part
of the characterization data, or be less sensitive than others to its inaccuracies. As shown in Fig.~\ref{fig:workflows}, a complete SOLEM workflow involves the execution of characterization circuits, used to obtain $\widetilde{\boldsymbol{\Lambda}}_{L|}$ in post-processing. Our circuit-level simulations in Sec.~\ref{sec:circ_leve_simul} simplify this workflow by starting with a given physical noise model; focusing on the task of mapping this to approximately-predictive logical channels $\widetilde{\boldsymbol{\Lambda}}_{L|}$ using P2LC. The solvable model analyzed in Sec.~\ref{sec:model} and the non-Clifford effective logical simulations of Sec.~\ref{Sec: effective logical simulation} further simplify by specifying simplified forms for  exact logical channels $\boldsymbol{\Lambda}_{L|}$.  We next describe three specific constructions within the SOLEM framework, starting with two constructions based on quasi-probability distributions.

\section{\gls{mname} with quasi-probability distributions}\label{sec:solem-qp-based}

\subsection{\gls{mname}-PEC}\label{sec:pec}
In this section we show how an \textit{unbiased} SOLEM estimator for the ideal expectation value $\E{O}$ can be constructed, given the exact predictive logical error channels of Sec.~\ref{subsec:logical error ch}, and assuming they satisfy a specific linear constraint. More generally, \textit{controlled} biases can result both from violations of this constraint, and from inaccurate characterization.

For reference, first consider the naive estimation of $\langle O\rangle$, using just EC. The estimator is the empirical mean $\mathsf{EC}=N^{-1}\sum_i o_i$, and its expectation is  
\begin{align}\label{eq:exp_val_of_ec}
    \mathbb{E}[\mathsf{EC}]
    &=N^{-1}\sum_i\mathbb{E}[\mathbb{E}[o_i| \bsk]] \\
    &=\mathbb{E}[\mathbb{E}[o| \bsk]] \nonumber\\
    &=\sum_{\bsk}p_{\bsk}\,
    \bbra{O}\Lambda_{L|\bsk}^{(D)}g_D\cdots \Lambda_{L|\bsk}^{(1)}g_1\kket{\psi} \nonumber\\
    &=\bbra{O}\Lambda_{L}^{(D)}g_D\cdots \Lambda_{L}^{(1)}g_1\kket{\psi}.\nonumber
\end{align}
This is just the faulty error-corrected expectation value, written in a purely logical form using exact (syndrome-averaged) channels $\Lambda_L^{(j)}$. Above we used the law of total expectation to pass through  
syndrome-class-conditioned expectation values and channels, though this isn't necessary for this simple estimator. Assuming $\|O\|_\infty\leq1$, the variance is just
\begin{align}
    \mathbb{V}[\mathsf{EC}]
    =N^{-1}\mathbb{V}[o]
    \leq N^{-1}.
\end{align}

\gls{mname}-PEC is designed to replace the naturally occurring probabilities $p_{\boldsymbol{k}}$ in Eq.~\eqref{eq:exp_val_of_ec} with chosen quasi-probabilities (QPs) $q_{\boldsymbol{k}}\in \mathbb{R}$,  $\sum_{\boldsymbol{k}}q_{\boldsymbol{k}}=1$, such that the conditioned logical errors \textit{cancel}, instead of being averaged. The SOLEM-PEC estimator is defined by 
\begin{align}
    \mathsf{Est}
    =\frac{1}{N}\sum_i \frac{q_{\bsk_i}}{p_{\bsk_i}}\, o_i,
\end{align}

where $q_{\bsk}=\prod_j q^{(j)}_{\bsk}$ is a product of layer-wise QPs
\begin{align}
    q^{(j)}_{\bsk}=q^{(j)}_{k_j|k_{1:j-1}},
\end{align}
in analogy with the structure of $p_{\bsk}$ (Eq.~\eqref{eq:pk}). The expectation is
\begin{align}\label{eq:exp_val_of_est}
    \mathbb{E}[\mathsf{Est}]
    &=N^{-1}\sum_i\mathbb{E}[(q_{\bsk_i}/p_{\bsk_i})\mathbb{E}[o| \bsk_i]] \\
    &=\mathbb{E}[(q_{\bsk}/p_{\bsk})\mathbb{E}[o| \bsk]] \nonumber\\
    &=\sum_{\bsk}q_{\bsk}\,
    \bbra{O}\Lambda_{L|\bsk}^{(D)}g_D\cdots \Lambda_{L|\bsk}^{(1)}g_1\kket{\psi} \nonumber\\
    &=\sum_{\bsk}
    \bbra{O}
    q_{\bsk}^{(D)}\Lambda_{L|\bsk}^{(D)}g_D\cdots
    q_{\bsk}^{(1)}\Lambda_{L|\bsk}^{(1)}g_1
    \kket{\psi},
    \nonumber
\end{align}
in analogy with Eq.~\eqref{eq:exp_val_of_ec}.

\begin{proposition}[Cancellation condition]\label{claim:unbiased_qp}
Suppose that for each layer $j=1,\dots,D$, and history $k_{1:j-1}$, the coefficients $q^{(j)}_{\bsk}$ satisfy
\begin{align}\label{eq:qk_for_identity}
    \sum_{k_j} q_{\bsk}^{(j)}\Lambda_{L|\bsk}^{(j)}
    =
    \sum_{k_j} q_{k_j|k_{1:j-1}}^{(j)}\Lambda_{L|k_{1:j}}^{(j)}
    =I.
\end{align}
Then the SOLEM-PEC estimator $\mathsf{Est}$ is unbiased,
\begin{align}
    \mathbb{E}[\mathsf{Est}]=\bbra{O}g_D\cdots g_1\kket{\psi}=\E{O}.
\end{align}
\end{proposition}

\begin{proof}
Starting from Eq.~\eqref{eq:exp_val_of_est}, we first sum over $k_D$. Eq.~\eqref{eq:qk_for_identity} gives
\(
    \sum_{k_D} q^{(D)}_{\bsk}\Lambda^{(D)}_{L|\bsk}=I
\),
eliminating the factor $q^{(D)}_{\bsk}\Lambda^{(D)}_{L|\bsk}$. Repeating the same argument for $k_{D-1},k_{D-2},\dots,k_1$ yields
\(
    \mathbb{E}[\mathsf{Est}]=\E{O}
\).
\end{proof}

\subsubsection{Bias and shot overhead}
Of course, exact cancellation is too strict of a requirement, and it suffices to require that 
\begin{align}\label{eq:min_norm}
    \delta^{(j)}_{|k_{1:j-1}}:= \left\| \sum_{k_j} q_{k_j|k_{1:j-1}}^{(j)}\Lambda_{L|k_{1:j}}^{(j)}-I\right\|
\end{align}
is small enough, for all layers $j$ and histories $k_{1:j-1}$, in order to control the resulting mitigation bias. Worst-case bounds on the bias are phrased in terms of diamond norms, while more refined estimates of the bias can be written in terms of Frobenius norms \cite[Appendix A]{aharonov2025reliablehighaccuracyerrormitigation}.

The shot overhead of the SOLEM-PEC estimator is bounded as 
\begin{align}
N\,\mathbb{V}[\mathsf{Est}]=\mathbb{V}[(q_{\bsk}/p_{\bsk})o]
    \leq \mathbb{E}[q^2_{\bsk}/p^2_{\bsk}]=\sum_{\bsk} q_{\bsk}^2/p_{\bsk},
    \label{eq:gamma_eq}
\end{align}
irrespective of whether the estimator is exactly unbiased or not. We therefore define the function 
\begin{align}
    \mathcal{D}(q;p)=\sum_{\bsk} q_{\bsk}^2/p_{\bsk},
\end{align}
which we wish to minimize when constructing $q_{\bsk}$, in addition to $\delta$. Note that this is essentially the Pearson $\chi^2$ divergence of $q$ and $p$, 
\begin{align}
    \mathcal{D}(q;p)-1=\sum_{\bsk} (q_{\bsk}-p_{\bsk})^2/p_{\bsk},
\end{align}
where we used $\sum_{\bsk}q_{\bsk}=1=\sum_{\bsk}p_{\bsk}$. Further note that $\mathcal{D}(q;p)$ does not decompose as a product over layers, due to history-dependence,
\begin{align}
    \mathcal{D}(q;p)=\sum_{k_{1:D}}\prod_{j=1}^D\frac{(q_{k_j|k_{1:j-1}}^{(j)})^2}{p^{(j)}_{k_j|k_{1:j-1}}}.
\end{align}
We can however bound it as
\begin{align}
    \mathcal{D}(q;p)\leq \prod_{j=1}^D \mathcal{D}^{(j)}, 
\end{align}
where $\mathcal{D}^{(j)}$ are per-layer bounds on all possible histories, 
\begin{align}
    \mathcal{D}^{(j)}
    :=&
    \max_{k_{1:j-1}}
    \mathcal{D}^{(j)}_{|k_{1:j-1}},\\
    \mathcal{D}^{(j)}_{|k_{1:j-1}}:=&\sum_{k_j} \frac{(q_{k_j|k_{1:j-1}}^{(j)})^2}{p^{(j)}_{k_j|k_{1:j-1}}}.\nonumber
\end{align}
This forms the basis for minimizing $\mathcal{D}^{(j)}_{|k_{1:j-1}}$ separately for each layer and history, as a means to control the shot overhead. 

Summarizing this discussion, our construction of QP distributions attempts to jointly minimize $\delta^{(j)}_{|k_{1:j-1}}$ and $\mathcal{D}^{(j)}_{|k_{1:j-1}}$, separately for each layer and history, as a means to minimize both bias and variance. This generally leads to families of solutions, corresponding to an inherent tradeoff between bias and variance. We refer to Sec.~\ref{sec:circ_leve_simul} for a concrete example.     

\subsubsection{Comparison with external PEC}  

We stress the different roles of $q_{\bsk}$ and $p_{\bsk}$ here, compared to standard PEC~\cite{Temme2017QEM,Endo2018QEM}. In the standard case, circuit modifications labeled by $\sigma$ are sampled, and one is free to choose the probability $p_{\sigma}$, for a given $q_{\sigma}$. The standard choice is obtained by minimizing $\mathcal{D}(q;p)$ over $p$, with the constraints $\sum_{\sigma} p_{\sigma} =1, p_{\sigma} \geq 0$; yielding $p_{\sigma} ^*=|q_{\sigma} |/\|q\|_1$ and $\mathcal{D}(q;p^*)=\|q\|_1^2$. For this reason, when constructing $q_{\sigma} $, one attempts to minimize the `QP norm' $\|q\|_1$, controlling the resulting shot overhead. In contrast, in SOLEM-PEC, the naturally-occurring probabilities $p_{\bsk}$ are fixed. We therefore attempt to minimize $\mathcal{D}(q;p)$ as part of the construction of $q_{\bsk}$, rather than the QP norm. 

The above probability-optimization of standard PEC shows that $\mathcal{D}(q;p)\geq \|q\|_1^2$, which is one feature working against SOLEM-PEC, compared with external PEC.  This does not establish an ordering between the overheads of SOLEM-PEC
and external PEC, since they use different sets of logical maps:
SOLEM combines syndrome-conditioned maps,
whereas external PEC combines modifications of
syndrome-averaged maps. Syndrome conditioning can identify classes with substantially
smaller logical errors, enabling cancellation with a lower shot overhead, despite using fixed probabilities.
Our simplified model
(Sec.~\ref{sec:model-overheads}, Fig.~\ref{fig:pec-lam})
and circuit-level simulations
(Sec.~\ref{sec:circ_leve_simul},
Fig.~\ref{fig:memory-simulation})
demonstrate regimes in which SOLEM-PEC significantly
outperforms external PEC, and performs comparably to SALEM-PEC.

\subsection{\gls{mname}-ZNE}\label{sec:solem-zne}
SOLEM-PEC synthesizes the identity channel by linearly combining the
naturally occurring syndrome-conditioned logical channels. The same
construction can instead synthesize a family of partially mitigated or
amplified channels, providing an offline analogue of ZNE. 
For the construction below, we assume exact characterization and
exact solutions of the QP equations;
an approximate solution introduces an additional
channel-synthesis bias.

For each location \(j\), choose a reference logical channel
\(\Lambda_{\mathrm{ref}}^{(j)}\) and coefficients
\(y_{\bsk}^{(j)}\) satisfying
\begin{align}\label{eq:pea_eq_for_qk}
\sum_{k_j}y_{\bsk}^{(j)}
\Lambda_{L|\bsk}^{(j)}
=\Lambda_{\mathrm{ref}}^{(j)}.
\end{align}
We normalize the noise parameter \(a\) so that \(a=1\) corresponds to
the reference channel and \(a=0\) to the identity. Combining the
coefficients in Eq.~\eqref{eq:pea_eq_for_qk} with the identity-channel
coefficients of Eq.~\eqref{eq:qk_for_identity} gives\footnote{More generally one can use any parametrization 
\(y_{\bsk}^{(j)}(a)
= y_{\bsk}^{(j)} +\varphi(a) (q_{\bsk}^{(j)} -y_{\bsk}^{(j)})\)
with \(\varphi(0)=1\) and \(\varphi(1)=0\).
}
\begin{align}\label{eq:solem-zne-local-coefficients}
y_{\bsk}^{(j)}(a)
=a\,y_{\bsk}^{(j)} +(1-a)q_{\bsk}^{(j)},
\end{align}
which synthesizes
\begin{align}\label{eq:solem-zne-channel-path}
\Lambda^{(j)}(a)
:=a\Lambda_{\mathrm{ref}}^{(j)}+(1-a)I.
\end{align}

For a reference channel admitting an affine decomposition into
nonidentity CPTP error channels,
\begin{align}
\Lambda_{\mathrm{ref}}^{(j)}
={}&
\left(1-\sum_{\me\neq I}p^{(j)}(\me)\right)I
+\sum_{\me\neq I}p^{(j)}(\me)\me,
\label{eq:error-channel-rescaling}
\end{align}
Eq.~\eqref{eq:solem-zne-channel-path} rescales
\(p^{(j)}(\me)\mapsto a p^{(j)}(\me)\). 
Values \(0<a<1\) reduce the noise, while admissible values \(a>1\)
amplify it.

Define
\(y_{\bsk}(a):=\prod_{j=1}^{D}y_{\bsk}^{(j)}(a)\). The
estimator
\begin{align}\label{eq:solem-zne-estimator}
\widehat{O}_a
=\frac{1}{N}\sum_{i=1}^{N}
\frac{y_{\bsk_i}(a)}{p_{\bsk_i}}o_i
\end{align}
is unbiased for the expectation value \(\E{O}_a\) of the circuit with
channels \(\Lambda^{(j)}(a)\). Crucially, \(a\) enters only through the
classical shot weights. The same measured records therefore yield
correlated estimates of \(\E{O}_a\) at any collection of admissible\footnote{We call \(a\) admissible when every \(\Lambda^{(j)}(a)\) is CPTP.}
noise levels, without modifying or rerunning the circuit.

The same records also provide derivative information. For any
nonnegative integer \(l\),
\begin{align}\label{eq:solem-zne-derivative-estimator}
\left.\frac{d^l\widehat{O}_a}{da^l}\right|_{a=a_0}
=\frac{1}{N}\sum_{i=1}^{N}
\left.\frac{d^l}{da^l}
\left(\frac{y_{\bsk_i}(a)}{p_{\bsk_i}}\right)
\right|_{a=a_0}o_i
\end{align}
is an unbiased estimator for
\(\left.d^l\E{O}_a/da^l\right|_{a=a_0}\). Explicit derivative
weights are given in App.~\ref{sec:deriv}.

To perform SOLEM-ZNE, we approximate the noise response by a
low-dimensional model 
\[\E{O}_a \approx f_{\boldsymbol{\theta}}(a), \]
fit
\(\boldsymbol{\theta}\) using any chosen set of values, derivatives, or
both, and report \(f_{\widehat{\boldsymbol{\theta}}}(0)\) as the
mitigated estimate. Because all fitted quantities are constructed from
the same data, their full covariance should be included in the fit.  
One covariance-aware choice is generalized least squares (GLS), which
weights the residuals by their inverse covariance matrix. 
If the noise response belongs to the fitted family, GLS using all shots
at every amplification point has asymptotic covariance no larger than
GLS obtained by dividing the same measurement records into independent
ensembles assigned to those points; see Obs.~\ref{obs:shared-shot-gls}.

The construction extends naturally to multiple reference channels.
If \(y_{\bsk,\nu}^{(j)}\) represents
\(\Lambda_{\mathrm{ref},\nu}^{(j)}\), then
\begin{align}
y_{\bsk}^{(j)}(\boldsymbol a)
&=q_{\bsk}^{(j)}
+\sum_{\nu=1}^{R}a_\nu
\bigl(y_{\bsk,\nu}^{(j)}-q_{\bsk}^{(j)}\bigr)
\end{align}
synthesizes the channel family
\begin{align}
\Lambda^{(j)}(\boldsymbol a)
&=I+\sum_{\nu=1}^{R}a_\nu
\bigl(\Lambda_{\mathrm{ref},\nu}^{(j)}-I\bigr).
\end{align}
The same measured records therefore enable multivariate extrapolation
to the noiseless point \(\boldsymbol a=\boldsymbol 0\).

SOLEM-ZNE differs from conventional ZNE in both the choice of noise
levels and their statistical cost. In conventional ZNE, noise levels
are chosen before data acquisition and implemented through circuit
modifications, with separate shot ensembles typically used at each
level. When additional noise can be implemented as a probabilistic
mixture of physical operations, amplification incurs no QP sampling
penalty.

In SOLEM-ZNE, the noise levels are chosen entirely in post-processing,
and their estimators are correlated because they use the same shots.
Both amplification and reduction are implemented by reweighting the
measured records, so amplification is not necessarily less costly than
reduction. These differences affect the optimal noise levels, as
illustrated in Sec.~\ref{sec:model_qp_based}.

\section{SOLEM with global regression (SOLEM-GR)}\label{sec:solem-gr}

The QP-based constructions of the previous section require the available
syndrome-conditioned logical channels to provide a sufficiently accurate,
well-conditioned representation of the desired channel. This need not hold,
for example when a small number of syndrome classes is used for a block that
encodes several logical qubits. Even when a useful QP representation exists,
additional assumptions on the observable response can yield a lower sampling
overhead. SOLEM-GR introduces such an assumption and estimates the ideal
expectation value by regression, without first constructing a QP
representation. Rather than deliberately varying the circuit noise, SOLEM-GR
uses the natural shot-to-shot variation in syndrome-conditioned logical noise
to extrapolate to the ideal expectation value without modifying the circuit.

Without additional structure, the response \(\mathbb E[o\mid\bsk]\) may
contain as many as \(K^D\) unrelated conditional means. SOLEM-GR reduces this
complexity by mapping each syndrome record to a scalar global decay factor,
\begin{equation}
Z_{\bsph}:[K]^D\longrightarrow\mathbb R,
\end{equation}
where \(\bsph\) denotes a small number of nuisance parameters governing the response to logical
errors. The map depends on characterized syndrome-conditioned logical
channels, which are left implicit. We normalize \(Z_{\bsph}\) so that its
noiseless continuation, obtained by setting the characterized error parameters
to zero, satisfies
\begin{equation}
Z_{\bsph}^{(0)}(\bsk)=1.
\end{equation}

The central modeling assumption of SOLEM-GR is that the coarse-grained
response is approximately multiplicative,
\begin{equation}\label{eq:solem-gr-decay-model}
\mathbb E[o\mid Z_{\bsph_*}=z]\simeq \theta_* z,
\qquad \theta_*\simeq\E{O}.
\end{equation}
Individual syndrome records with the same value of \(Z_{\bsph_*}\) need not
have identical conditional responses; Eq.~\eqref{eq:solem-gr-decay-model}
models their average. Because the noiseless continuation of the decay factor
is one, the corresponding continuation of the fitted response is simply
\(\theta\), which is therefore the SOLEM-GR estimate of the ideal expectation
value.

For a loss function \(\ell:\mathbb R^2\to\mathbb R_+\), we fit the amplitude
\(\theta\) and nuisance parameters \(\bsph\) by
\begin{equation}\label{eq:solem-gr-estimator}
(\widehat\theta,\widehat\bsph)
\in\argmin_{\theta,\bsph}
\frac{1}{N}\sum_{i=1}^N
\ell\!\left(o_i,\theta Z_{\bsph}(\bsk_i)\right),
\qquad
\mathsf{Est}=\widehat\theta.
\end{equation}
Common choices are squared loss and, for binary outcomes
\(o\in\{-1,1\}\), the conditional Bernoulli negative log-likelihood
\begin{equation}
\ell(o,f)=-\log\!\left(\frac{1+of}{2}\right),
\qquad |f|<1.
\end{equation}
The target \(\theta\) must be identifiable from the distribution of
\((\bsk,o)\) under this model; otherwise the ideal expectation value cannot
be inferred from the measured data. The nuisance parameters need not be
uniquely identifiable if \(\theta\) is. A misspecified decay-factor model can
also produce a systematic error. We examine this source of bias numerically
in Sec.~\ref{Sec: effective logical simulation}.

The regression model can
be exact when the characterized channels imply a low-dimensional
conditional response. The following elementary example illustrates this
case.
\begin{example}[Global depolarization]\label{ex:glob_depol}
Suppose that, conditioned on \(\bsk\), the logical noise at location
\(j\) is globally depolarizing,
\begin{equation}
\Lambda^{(j)}_{L|\bsk}(\sigma)
=
(1-\epsilon^{(j)}_{\bsk})\sigma
+\epsilon^{(j)}_{\bsk}2^{-n}I\operatorname{Tr}(\sigma).
\end{equation}
For unitary logical layers and a traceless observable, global
depolarization attenuates the ideal expectation multiplicatively, giving
\begin{equation}
\mathbb E[o\mid\bsk]
=
\E{O}\prod_j(1-\epsilon^{(j)}_{\bsk}).
\end{equation}
Therefore, choosing
\[
Z(\bsk)=\prod_j(1-\epsilon^{(j)}_{\bsk}),
\]
gives the exact decay model
\(\mathbb E[o\mid Z=z]=\E{O}z\), with noiseless continuation
\(Z^{(0)}=1\). The squared-loss estimator is
\begin{equation}
\widehat\theta
=
\frac{\sum_i z_i o_i}{\sum_i z_i^2},
\end{equation}
and is unbiased.
\end{example}

This example shows that compression of the complete syndrome record to a
single feature need not lose information about the ideal expectation.
The solvable model of Sec.~\ref{sec:model} provides a less restrictive
setting in which the response contains an additional nuisance parameter.

\section{Circuit-level memory simulation}
\label{sec:circ_leve_simul}

To assess SOLEM with realistic syndrome statistics and temporal
correlations, we apply it to data from a circuit-level simulation of a
single-logical-qubit Steane-code memory.

We use the 1-FT flagged syndrome-extraction scheme of
Ref.~\cite{Chao_2018}, with one measurement ancilla and one flag ancilla;
see Fig.~\ref{fig:memory-simulation}(a). The six stabilizer measurement outcomes and
seven possible flag-register values (no flag or one of the six flagged
gadgets) give \(7\times2^6=448\) possible syndrome records per EC cycle.
Recovery is performed using a flag-aware lookup-table decoder.

Our circuit-level noise model applies the two-qubit depolarizing channel
\(
\sigma\mapsto
(1-\epsilon)\sigma+\epsilon\frac{I_4}{4}\tr(\sigma)
\)
after every \(\mathrm{CX}\) gate, and applies a bit-flip channel with error
probability \(\epsilon/2\) after every reset and before every measurement.
We set \(\epsilon=4\times10^{-4}\). For each of 30 independently sampled
depths \(D\), ranging from 256 to 20480 EC cycles, we generate
\(N=10^5\) shots. At each depth, all estimators are constructed from the
same syndrome and logical-measurement records.

\begin{figure*}[t!]
    \begin{centering}
    \includegraphics[width=\textwidth]{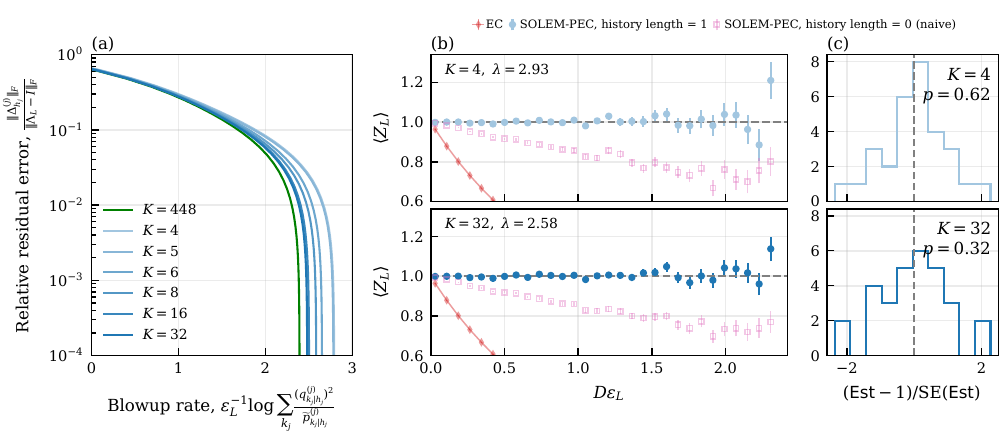}
    \end{centering}
    \caption{
\textbf{SOLEM-PEC bias--overhead tradeoff and mitigation performance.}
\textbf{(a)} Relative residual-channel Frobenius norm versus the local
sampling-overhead exponent for syndrome partitions of size \(K\), including
the fine-grained \(K=448\) partition. We define
\(\Delta^{(j)}_{h_j}:=
\sum_{k_j}q^{(j)}_{k_j\mid h_j}
\widetilde{\Lambda}^{(j)}_{L\mid k_j,h_j}-I\).
The curves use channels characterized with no incoming logical error,
corresponding to \(h_j=()\), and their points correspond to the sampled
values of \(\eta\) in Eq.~\eqref{eq:qp_cost_func}. For every displayed
partition, the rightmost point has a negligible residual norm.
\textbf{(b)} SOLEM-PEC estimates of \(\langle Z_L\rangle\) versus
\(D\epsilon_L\) for \(K=4\) (top) and \(K=32\) (bottom). Circles use
history-one channels, \(h_j=(k_{j-1})\), whereas faint pink squares use the
naive history-zero P2LC approximation, \(h_j=()\), which neglects incoming
logical errors to each EC cycle. For every history, we use the sampled QP
solution with the smallest residual Frobenius norm. The unmitigated EC
estimates are shown for comparison.
\textbf{(c)} 
Histograms of standardized residuals for the history-one
estimators over all 30 independently sampled circuit depths.
We test consistency with the ideal value using a standard chi-squared test
and report the corresponding \(p\)-values, which reveal no statistically
significant bias.
\label{fig:solem_pec_panels}
}
\end{figure*}

We use P2LC to obtain approximate syndrome-conditioned logical channels.
Specifically, we replace the exact channels
\(\Lambda^{(j)}_{L\mid\bsk}\) by truncated channels
\(\widetilde{\Lambda}^{(j)}_{L\mid k_j,h_j}\) with history length one.
For \(j>1\), the local record is \((k_{j-1},k_j)\), with
\(h_j=(k_{j-1})\); for \(j=1\), no previous syndrome is available and
\(h_1=()\).

The first step in the P2LC procedure is to partition the original syndromes
\(s\) into \(K\) subsets. We construct the partition using the
syndrome-conditioned logical Pauli channels \(\Lambda_{L\mid s}\)
characterized without incoming errors. The 96 syndromes with the smallest
nonidentity logical-error probabilities form subset \(k=0\). The remaining
syndromes are clustered according to the three nonidentity Pauli
coefficients of \(\Lambda_{L\mid s}\), using syndrome-probability-weighted
\(k\)-means. The size 96 of the first subset is a tuned hyperparameter.

\vspace{5pt}

\paragraph*{SOLEM-PEC and SOLEM-ZNE.}
For SOLEM-PEC, we seek accurate local cancellation for each history \(h_j\)
while controlling the sampling overhead. We therefore minimize a weighted
sum of the sampling second moment and the squared residual-channel norm,
\begin{align}
\label{eq:qp_cost_func}
\min_{\{q^{(j)}_{k_j\mid h_j}\}_{k_j}}
\Bigg\{
\sum_{k_j}
\frac{\bigl(q^{(j)}_{k_j\mid h_j}\bigr)^2}
     {\widetilde{p}^{(j)}_{k_j\mid h_j}}
+
\eta
\left\lVert
\sum_{k_j}
q^{(j)}_{k_j\mid h_j}
\widetilde{\Lambda}^{(j)}_{L\mid k_j,h_j}
-I
\right\rVert_F^2
\Bigg\},
\end{align}
where \(\eta>0\) controls the bias--overhead tradeoff. The results in
Fig.~\ref{fig:solem_pec_panels} illustrate this tradeoff and show the
importance of retaining the previous syndrome in the P2LC model. For the
history-one estimators, the consistency tests give \(p=0.62\) for \(K=4\)
and \(p=0.32\) for \(K=32\); thus, no bias is statistically resolved.

Figure~\ref{fig:memory-simulation}(d) also shows two-point SOLEM-ZNE with
\(K=4\) and history length one, for which the same test gives \(p=0.22\).
For SOLEM-PEC, increasing the partition size to \(K=32\) gives
\(\lambda=2.58\), close to the value 2.46 obtained for
fine-grained SALEM in Ref.~\cite{aharonov2025syndromeawaremitigationlogical}.

\vspace{5pt}

\paragraph*{SOLEM-GR.}
For a memory circuit with Pauli channels
\(\widetilde{\Lambda}^{(j)}_{L\mid k_j,h_j}\) and a Pauli observable, the
truncated P2LC model gives
\[
\mathbb{E}[o\mid\bsk]
=
\E{O}\prod_{j=1}^{D}
\left(1-2\rho_{k_j,h_j}\epsilon_{L\mid k_j,h_j}\right).
\]
Here \(\epsilon_{L\mid k_j,h_j}\) is the nonidentity Pauli probability in
\(\widetilde{\Lambda}^{(j)}_{L\mid k_j,h_j}\), and
\(\rho_{k_j,h_j}\) is the fraction of that probability that flips the
observable's sign. For \(Z_L\), it is the total \(X\)- and \(Y\)-error
probability divided by \(\epsilon_{L\mid k_j,h_j}\).

At a fixed depth \(D\), independently fitting all
\(\rho_{k,h}\) makes \(\E{O}\) structurally non-identifiable, rather than
merely increasing its estimation variance. Indeed, for sufficiently small
\(c\), the transformation
\(\E{O}\mapsto e^{-Dc}\E{O}\) and
\(\rho_{k,h}\mapsto
[1-e^c(1-2\rho_{k,h}\epsilon_{L\mid k,h})]/
(2\epsilon_{L\mid k,h})\)
leaves every conditional response invariant. We therefore use 
\begin{align*}
Z_{\phi}(\bsk)
&=
\prod_{j=1}^{D}
\left(1-\phi\epsilon_{L\mid k_j,h_j}\right),
\qquad 
\mathbb{E}[o \mid Z_\phi = z ] \simeq \theta z, 
\end{align*}
and jointly fit \(\theta, \phi\in\mathbb{R}\) by unweighted least squares.
A chi-squared goodness-of-fit test gives \(p=0.78\), providing no statistically significant evidence of bias.

\section{Analysis of a solvable model}\label{sec:model}
To gain qualitative insight into the overheads of the SOLEM variants and compare them with existing mitigation methods, we introduce an effective logical-noise model that abstracts away the details of particular codes, P2LCs, and decoding algorithms. We first define the model and derive its basic properties. We then analyze the SOLEM estimators within this model and compare their sampling overheads with those of existing methods.

\subsection{Model definition and structure}
Consider an arbitrary circuit with \(V\) designated sites at which syndromes
are measured. The sites may occur at arbitrary circuit locations. We impose
the following assumptions:
\begin{enumerate}[label=(\roman*)]
\item each syndrome is assigned a binary label \(k_j\in\{0,1\}\), giving
the record \(\bsk=(k_1,\ldots,k_V)\);
\item the syndrome-conditioned logical channel is independent of the site, and has the parallel-channel form
\begin{align}\label{eq:parallel}
\Lambda_{L|k}=(1-\epsilon_k)I+\epsilon_k E,
\end{align}
where \(E\) is CPTP and \(0<\epsilon_0<\epsilon_1<1\); and
\item the syndromes are independent and identically distributed, with
\(\mathbb P(k_j=1)=p_1\), where \(0<p_1<1\).
\end{enumerate}
The two conditional channels therefore lie along the common direction
\(E-I\).\footnote{Section~\ref{sec:non-parallel-Es} discusses the
nonparallel case.} For one-qubit depolarizing noise, for example,
\(E=(X+Y+Z)/3\).

Let \(N_1:=|\boldsymbol{k}|\) denote the number of syndrome-1
outcomes. The mean error probability is
\begin{align}
\epsilon_L := p_1 \epsilon_1 + (1-p_1) \epsilon_0,
\end{align}
and the corresponding mean channel is
\(\Lambda_{L}=(1-\epsilon_L)I+\epsilon_L E\).

In this model, error-count statistics and the circuit response
separate as follows.
For \(J\subseteq[V]\), let \(o_J\)
denote the expectation value of \(O\) when \(E\) acts at the locations in
\(J\) and \(I\) acts at all other noisy sites. The resulting decomposition
is stated below.
\begin{proposition}[Parallel-channel decomposition]
\label{prop:ev-cond-N1}
For \(l=0,\ldots,V\),
\[
  \mathbb{E}\!\left[o\mid N_1=l\right]
  =\sum_{r=0}^{V}K_{r,l}C_r,
\]
where
\[
  C_r:=\frac{1}{\binom{V}{r}}
  \sum_{\substack{J\subseteq[V]\\|J|=r}}o_J,
  \qquad
  K_{r,l}:=\mathbb P(N_E=r\mid N_1=l),
\]
and \(N_E\) is the number of locations at which an error \(E\) occurs.
Conditioned on \(N_1=l\), \(N_E\) is distributed as the sum of
independent \(\operatorname{Bin}(l,\epsilon_1)\) and
\(\operatorname{Bin}(V-l,\epsilon_0)\) random variables.
\end{proposition}
\begin{proof}
The proof is given in App.~\ref{app:proof-ev-cond-N1}.
\end{proof}

The decomposition separates the error-count statistics, encoded by
\(K_{r,l}\), from the circuit response, encoded by \(C_r\), with
\(C_0=\E{O}\). A particularly tractable response arises when \(C_r\) is
exponential in \(r\), as characterized by the following proposition.
\begin{proposition}[Exponential \(C_r\)]\label{prop:exp_Cr}
The following are equivalent:
\begin{enumerate}
  \item[(i)] \(C_r=C_0(1-\rho)^r\) for \(r=0,\dots,V\);
  \item[(ii)] for some pair
  \(0<\epsilon_0<\epsilon_1<1\), equivalently for every such pair,
  \[
  \mathbb E[o\mid N_1=l]
  =C_0(1-\rho\epsilon_1)^l(1-\rho\epsilon_0)^{V-l},
  \qquad l=0,\dots,V;
  \]
  \item[(iii)]
  \(\mathbb E[o]=C_0(1-\rho\epsilon_L)^V\) for every
  \(\epsilon_L\) in some nonempty open interval, equivalently for all
  \(\epsilon_L\).
\end{enumerate}
\end{proposition}
\begin{proof}
The proof is given in App.~\ref{app:proof_exp_cr}.
\end{proof}

Under the equivalent conditions of Prop.~\ref{prop:exp_Cr}, rescaling the
reference channel according to
\((1-\epsilon_0)I+\epsilon_0E\mapsto
(1-a\epsilon_0)I+a\epsilon_0E\) gives
\[
\mathbb{E}[o]=\E{O}(1-\rho a \epsilon_0)^V,
\]
which supplies the response model used for ZNE.
The factorized conditional response in condition~(ii) similarly supplies
the SOLEM-GR model of Sec.~\ref{sec:model-solem-gr}.
The conditions of Prop.~\ref{prop:exp_Cr} hold exactly in some cases, 
e.g. for global depolarizing noise (see Ex.~\ref{ex:glob_depol}).

Importantly, the exponential form of \(C_r\) is not merely an analytically convenient
assumption: for certain circuit families and error channels \(E\), the
coefficients \(C_r\) are well approximated by an exponential over the relevant range of \(r\), as illustrated in
Fig.~\ref{fig:cr_in_floquet}.

\begin{figure}[t!]
    \begin{centering}
    \includegraphics[width=\columnwidth]{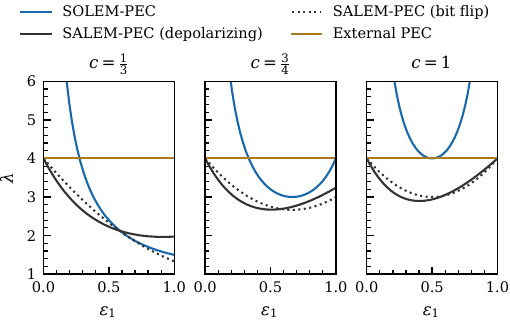}
    \end{centering}
    \caption{
    Comparison of the leading shot-overhead exponents
    \(\lambda\) of SOLEM-PEC, SALEM-PEC, and external PEC as functions of
    \(\epsilon_1\). The SALEM-PEC curves correspond to single-qubit
    depolarizing and bit-flip channels.
    Each panel fixes \(c\) and considers binary syndrome
    partitions satisfying \(p_{1|L}=1-c\epsilon_1\), or equivalently
    \(\epsilon_0/\epsilon_1=c\epsilon_L/(1-p_1)
    =c\epsilon_L+O(\epsilon_Lp_1)\). In general, the relation among these
    parameters is determined by the circuit, noise model, and syndrome
    classifier. The linear relation used here is a simplified model motivated
    by the empirical behavior observed for Steane and surface codes in
    Ref.~\cite{aharonov2025syndromeawaremitigationlogical}. 
    \label{fig:pec-lam}
    }
\end{figure}

\begin{figure*}[t!]
    \centering
    \includegraphics[width=\textwidth]{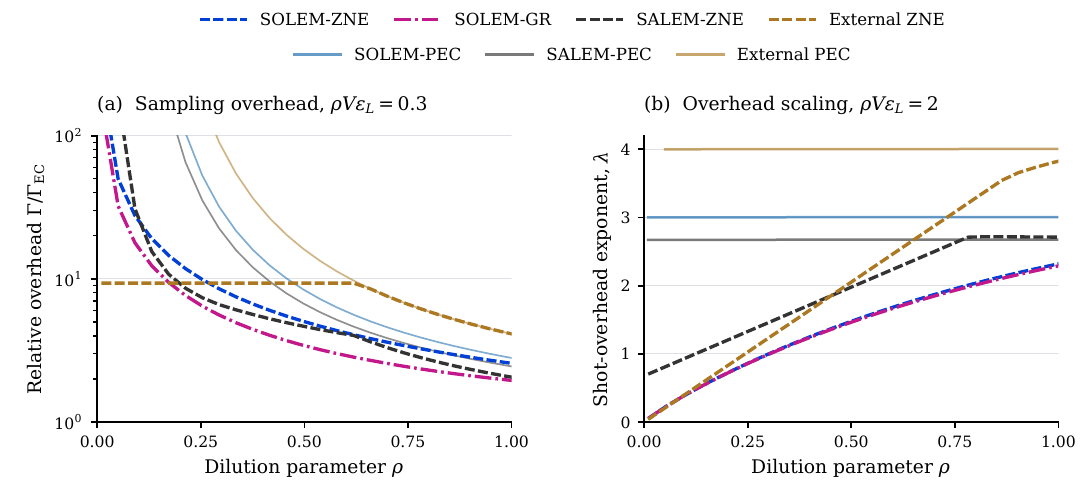}
    \caption{
        \textbf{Sampling overheads and scaling exponents of SOLEM, SALEM, and external error mitigation in the solvable model.}
        The comparison uses the binary parallel-channel model of Sec.~\ref{sec:model} and assumes the exponential
        observable response of Prop.~\ref{prop:exp_Cr}.
        \textbf{(a)} Relative sampling overhead
        $\Gamma/\Gamma_{\mathrm{EC}}$
        versus the dilution parameter $\rho$, at
        fixed $\rho V\epsilon_L=0.3$.
        Thin solid curves show PEC methods; thicker dashed and dash-dotted
        curves show ZNE and SOLEM-GR, respectively.
        \textbf{(b)}
        Effective shot-overhead exponent        $\lambda=\log(\Gamma/\Gamma_{\mathrm{EC}})/(V\epsilon_L)$,
        at
        $\rho V\epsilon_L=2$.
        SOLEM-ZNE and SOLEM-GR nearly coincide in this panel.
        In both panels, $V$ varies with $\rho$ to maintain the indicated
        value of $\rho V\epsilon_L$, and each method's overhead is minimized
        over the sampled syndrome partitions.
        The remaining parameters are fixed at $\langle O\rangle=0.8$ and
        $\epsilon_L=10^{-4}$. For each method and value of $\rho$, the plotted
        value is the smallest
        overhead among 15 sampled binary partitions with
        $p_{1|L} = 1 -\frac{3}4{}\epsilon_1$,
        equally spaced $\epsilon_1$ values from
        $1.10\times10^{-3}$ to $0.7$. 
        SALEM-PEC uses the QP norm for single-qubit
        depolarizing channels given in Table~\ref{tab:overhead-summary}.
        \label{fig:solem-zne-overheads}
    }
\end{figure*}

\subsection{Overheads of SOLEM variants in this model}\label{sec:model-overheads}
Under the assumptions of Prop.~\ref{prop:exp_Cr}, SOLEM-PEC is unbiased, while
the response models used by SOLEM-ZNE and SOLEM-GR are correctly
specified, making the corresponding estimators asymptotically unbiased.
The analysis therefore focuses on their sampling overheads. The leading-order
results are summarized in Tab.~\ref{tab:overhead-summary}.

\subsubsection{\gls{mname} with QP distributions}\label{sec:model_qp_based}
Equation~\eqref{eq:pea_eq_for_qk} admits a closed-form solution in this
model. Because the channels are history independent,
$q^{(j)}_{k_j|k_{1:j-1}}$ depends only on \(k_j\); we denote it by
\(q_{k_j}\). Taking \(\Lambda_{L|0}\) as the reference channel, define
\(\Lambda_a:=(1-a\epsilon_0)I+a\epsilon_0E\) and solve
$q_0\Lambda_{L|0}+q_1\Lambda_{L|1}=\Lambda_a$. The solution is
\begin{equation}\label{eq:q_sol_parallel}
  q_1(a):=\frac{\epsilon_0 / \epsilon_1 (a-1)}{1-\epsilon_0/\epsilon_1},
  \qquad
  q_0(a):=\frac{1-a\epsilon_0/\epsilon_1}
                 {1-\epsilon_0/\epsilon_1}.
\end{equation}
The solution is independent of the structure of \(E\).
For the overhead analysis, assume additionally that \(O^2=I\), so that each measurement outcome is in \(\{-1,1\}\), and restrict \(0<\rho<2\).
We have
\begin{align*}
\sum_{\bsk}\frac{q_{\bsk}^{2}}{p_{\bsk}} 
= & \left(\frac{q_{1}^{2}}{p_{1}}+\frac{q_{0}^{2}}{1-p_{1}}\right)^{V} \\
=& \left[1 +  \frac{\lambda_a \epsilon_L}{(1-p_1)(1 -\epsilon_0/\epsilon_1)^2} \right]^V
\end{align*}
where
\begin{align}
\lambda_{a}  
& := p_{1|L}^{-1} \epsilon_1^{-1} \left( 1-   \frac{a}{a_{EC}} \right)^2,
\end{align}
with $a_{EC}:=\epsilon_L/\epsilon_0$, the amplification corresponding to the EC computation itself,
and $ p_{1|L} = p_1 \epsilon_1 / \epsilon_L$.
The results above are exact within the solvable model. Compact closed-form
expressions follow in the regime
\[
0<\epsilon_0\ll\epsilon_1<1,\qquad p_1\ll1.
\]
By Eq.~\eqref{eq:gamma_eq} the shot overhead corresponding to $a$ is
\begin{align}\label{eq:gamma_a_solem}
\Gamma_{a}  
& = e^{\lambda_a \epsilon_LV }(1 + O(V \lambda_a \epsilon_L[p_1+\epsilon_0/\epsilon_1 + \lambda_a \epsilon_L] )) 
- \E{O}^2_a, 
\end{align}
where $\E{O}_a = \E{O} (1-\rho a \epsilon_0)^V$.
The first term of \(\Gamma_a\) is symmetric about \(a_{EC}\); hence,
increasing and decreasing the noise incur the same reweighting cost. By
contrast, conventional noise amplification by circuit modification does not
itself introduce a sampling cost.

\textit{SOLEM-PEC.}
Complete cancellation corresponds to $a=0$. Equation~\eqref{eq:q_sol_parallel}
then gives
\begin{align*}
q_{1} &= -\epsilon_0 / \epsilon_1 + O((\epsilon_0/\epsilon_1)^2), \\
q_0  &= 1 +\epsilon_0 / \epsilon_1  + O((\epsilon_0/\epsilon_1)^2),
\end{align*}
and hence the blowup rate
\begin{align}\label{eq:solem_qp_lam_lem}
\lambda = \epsilon_1 ^{-1} p^{-1}_{1|L}.
\end{align}

\textit{Two-point SOLEM-ZNE.}
Here we use two distinct values \(0< a_1<a_2\), with correlated
estimators \(\hat{o}_{a_1}\) and \(\hat{o}_{a_2}\) constructed from the
same shots, and covariance matrix
\(\Omega_{j,k} := N \operatorname{Cov}(\widehat o_{a_j},\widehat o_{a_k})\). 
For \(\boldsymbol{e}_{\bst}\in\mathbb R^2\) defined by
\[ 
(\boldsymbol e_{\bst})_i
=\hat{o}_{a_i}-\theta_1(1-\theta_2a_i\epsilon_0)^V,
\] 
the generalized nonlinear least-squares estimator is
\begin{align}
\hat{\bst}(a_1,a_2)
=\arg\min_{\bst}
\boldsymbol e_{\bst}^{\mathsf T}\Omega^{-1}\boldsymbol e_{\bst}.
\end{align}
The amplification pair is then chosen in post-processing to minimize
\(\V[\hat{\theta}_1(a_1,a_2)]\).
As shown in App.~\ref{sec:app_solem_zne}, the shared-shot covariance can only reduce the variance relative to the optimal independent-shot allocation.

Numerical minimization indicates that the variance infimum is approached
as the two amplification factors coalesce at
\begin{align}
a_c=a_{\mathrm{EC}}(\bar{\lambda}-1),
\qquad
\bar{\lambda}=2-\rho\epsilon_1p_{1|L}.
\end{align}
The variance remains finite in this limit, as stated below.
\begin{proposition}[Coalescing-amplification limit]
\label{prop:2pt_zne_limit}
The correlated two-point estimator satisfies
\begin{align}
&\lim_{\substack{a_1,a_2\to a_c\\ a_1<a_2}}
N\V\!\left[\widehat{\theta}_1(a_1,a_2)\right]
\nonumber\\
&\qquad=
e^{\bar{\lambda}\rho V\epsilon_L}
\left[
1+\frac{V\epsilon_L}{p_{1|L}\epsilon_1}
(\bar{\lambda}-1)^2
\right]
-\E{O}^2
\end{align}
to leading order.
\end{proposition}
\begin{proof}
See App.~\ref{sec:app_solem_zne}.
\end{proof}
This solution is equivalent to SOLEM-ZNE with a point and first derivative 
at \(a_c\).
Accordingly, two-point SOLEM-ZNE is evaluated at two distinct amplification
factors chosen arbitrarily close to \(a_c\), and we use the coalescing-limit
variance in Prop.~\ref{prop:2pt_zne_limit} as its asymptotic overhead.

\subsubsection{SOLEM-GR}\label{sec:model-solem-gr}
Proposition~\ref{prop:exp_Cr} gives
\[
\mathbb E[o\mid N_1=l]
  = \E{O} (1 - \rho\epsilon_1)^l(1-\rho\epsilon_0)^{V-l}.
\]
This is the regression function induced by the coarse-grained feature
\(N_1=|\bsk|\): individual records of the same weight may have different
conditional responses, while their average has the factorized form above.
Accordingly, 
\begin{align}\label{eq:ftheta}
    Z_{\phi}(\bsk) &= 
    (1- \phi \epsilon_1)^{|\bsk|} (1- \phi \epsilon_0)^{V-|\bsk|},
\end{align}
with \(\mathbb{E}[o \mid Z_{\phi} =z] = \theta z\), 
and \( \theta \in [-1,1], \, \phi \in [0,2 ]\).

\begin{proposition}[Fisher-information bound]\label{prop:gr-fisher}
Consider the conditional Bernoulli model associated with
Eq.~\eqref{eq:ftheta}, with
\(\bst=(\theta ,\phi)\), where
\(\bst_0=(\E{O}, \rho)\).
Treating \(\bst\) as unknown, every regular, locally asymptotically
unbiased estimator \(\widehat{\theta}_{1,N}\) whose information about
\(\bst\) is obtained solely from
\(\{(|\bsk_i|,o_i)\}_{i=1}^{N}\) satisfies
\[
\liminf_{N\to\infty}
N\V_{\bst_0}\!\left(\widehat{\theta}_{1,N}\right)
\geq \Gamma_{\rm F},
\]
where, to leading order in \(p_1\), \(\epsilon_0/\epsilon_1\), and the
squared conditional means,
\begin{align*}
\Gamma_{ \rm F}
={}&
e^{\bar{\lambda}\rho V\epsilon_L}
\left[
1+\frac{V\epsilon_L}{p_{1|L}\epsilon_1}
(\bar{\lambda}-1)^2
\right]
\\
&-\E{O}^2
e^{(2-\bar{\lambda})\rho V\epsilon_L
(2-\rho\epsilon_1)^2}F_0,
\end{align*}
where \(F_0\) is given in Eq.~\eqref{eq:gr-F0}.

Under the standard regularity and identifiability conditions, the
maximum-likelihood estimator asymptotically attains the exact bound.
The complete expression is given in App.~\ref{sec:fisher}.
\end{proposition}
When the conditional means are uniformly small, the asymptotic variance
of unweighted least squares agrees with the Fisher-information bound
through quadratic order; their difference is fourth order in
\(\max_{0\leq l\leq V}
|\E{O}(1-\rho\epsilon_0)^{V-l}(1-\rho\epsilon_1)^l|\)
(see App.~\ref{sec:app_solem_gr}).

\subsection{Comparison of shot overheads}
As seen in Tab.~\ref{tab:overhead-summary}, the overheads of all SOLEM
variants depend explicitly on the mitigated volume \(V\), rather than
only through the effective decay strength \(\rho V\epsilon_L\), as is
the case for ExtZNE. The SALEM-ZNE construction considered here has a
similar dependence because it rejects shots containing nonzero
syndrome outcomes, contributing a factor
\(\exp(p_1V)=\exp(p_{1|L}\epsilon_1^{-1}\epsilon_LV)\) to the sampling
overhead.

Figure~\ref{fig:pec-lam} compares the sampling overheads of SOLEM-PEC,
SALEM-PEC, and ExtPEC. SOLEM-PEC can outperform SALEM-PEC,
depending on the relation between \(\epsilon_1\) and \(p_{1|L}\).

A broader comparison, including SOLEM-ZNE and SOLEM-GR, is shown in
Fig.~\ref{fig:solem-zne-overheads}. For ExtZNE and SALEM-ZNE, the
optimization permits either noise amplification or noise reduction,
with the optimal choice depending on \(\rho\). Because the SOLEM
overheads depend explicitly on \(V\), ExtZNE performs best at small
\(\rho\) in the moderate-decay regime shown in
Fig.~\ref{fig:solem-zne-overheads}(a). In the stronger-decay regime of
Fig.~\ref{fig:solem-zne-overheads}(b), SOLEM-ZNE and SOLEM-GR have
nearly identical overheads and outperform the other methods considered.

\subsection{Remark on the case of nonparallel conditional errors}
\label{sec:non-parallel-Es}
The error-count decomposition extends to nonparallel conditional error
channels of the form
\[
\Lambda_{L|k}=(1-\epsilon_k)I+\epsilon_kE_k,
\]
with
\(k\in\{0,1\}\). In this case, one can show that
\[
\mathbb E[o\mid N_1=l]
=
\sum_{r_1=0}^{l}\sum_{r_0=0}^{V-l}
K_{r_0,r_1,l}C_{r_0,r_1},
\]
where
\(K_{r_0,r_1,l}:=
\mathbb P(N_{E_0}=r_0,N_{E_1}=r_1\mid N_1=l)\)
contains the joint error-count statistics, and
\[
C_{r_0,r_1}
:=
\frac{1}{\binom{V}{r_1}\binom{V-r_1}{r_0}}
\sum_{\substack{
J_0,J_1\subseteq[V]\\
J_0\cap J_1=\varnothing\\
|J_0|=r_0,\ |J_1|=r_1
}}
o_{J_0,J_1}
\]
is the circuit response averaged over the locations of the two error
types. Here, \(o_{J_0,J_1}\) is the expectation value when \(E_0\) and
\(E_1\) act at the locations in \(J_0\) and \(J_1\), respectively.

An analogue of the exponential-\(C_r\)
condition is
\(C_{r_0,r_1}=\E{O}(1-\rho_0)^{r_0}(1-\rho_1)^{r_1}\), 
which implies
\[
\mathbb E[o\mid N_1=l]
=
\E{O}(1-\rho_0\epsilon_0)^{V-l}
(1-\rho_1\epsilon_1)^l.
\]
At fixed \(V\), however, the three parameters
\(\E{O},\rho_0,\rho_1\) are not jointly identifiable from this response.
An additional relation between \(\rho_0\) and \(\rho_1\) is therefore
required to fit the ideal expectation value.

\begin{figure*}[t!]
    \begin{centering}
    \includegraphics[width=0.8\textwidth]{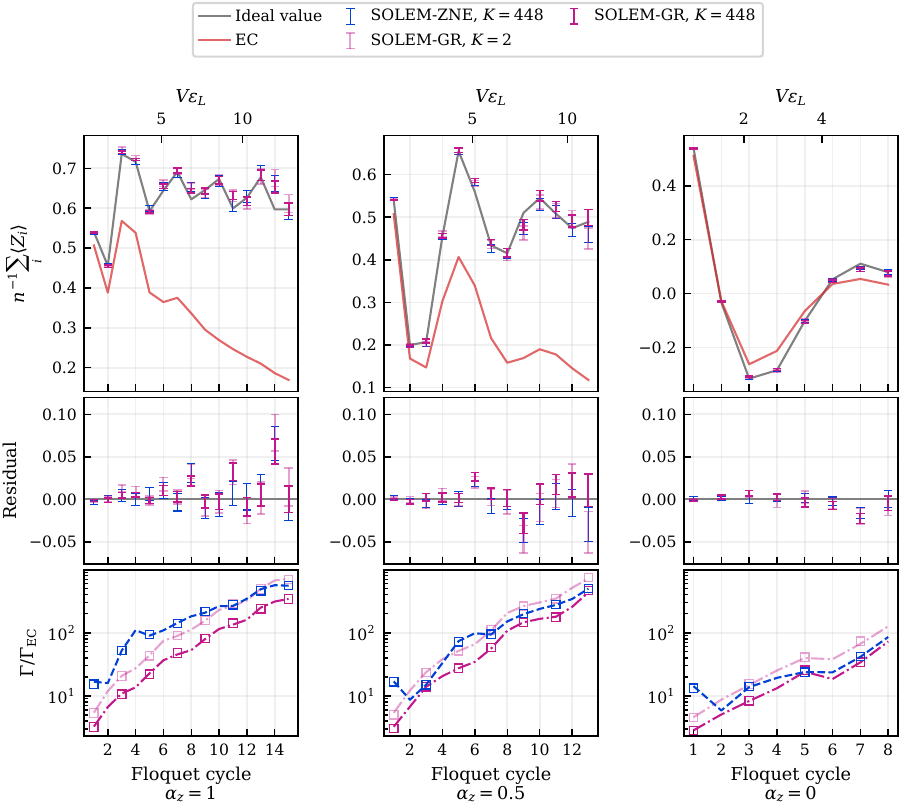}
    \end{centering}
    \caption{
    Results of the effective logical-channel Floquet simulations.
    We set \(\alpha_{zz}=\pi/6\) and \(\alpha_x=1\); the columns correspond
    to different values of \(\alpha_z\).
    For \(\alpha_z=0\), the \(Z\)-rotation layer is omitted, so each Floquet
    cycle contains five rather than six layers.
    The mean logical infidelity is
    \(\epsilon_L=8.9\times10^{-3}\), and each circuit is sampled using
    \(N=10^5\) shots.
    In the bottom row, lines show the analytical asymptotic-variance
    predictions, while open squares show bootstrap estimates obtained from
    \(2000\) resamples.
    In this simulation, the amplification of noise is exact so that biases of SOLEM-ZNE
    are due only to the misspecification of the response function \(f_{\bst}(a)\).
    For most cycle numbers, the SOLEM-GR and SOLEM-ZNE estimates agree with
    the ideal value within statistical uncertainty. When visible deviations
    occur, they are strongly correlated between the two methods, and we do
    not resolve an additional bias specific to the SOLEM-GR heuristic.
    For SOLEM-ZNE, we select the pair of amplification factors that minimizes
    the predicted estimation variance and hold this pair fixed throughout
    the bootstrap procedure. The reported bootstrap variance is therefore
    conditional on the selected pair and does not include its selection
    uncertainty.
    \label{fig:flouqet_sim}    
    }
\end{figure*}

\section{Effective logical-channel simulation of non-Clifford circuits}
\label{Sec: effective logical simulation}
The simulations in this section study the bias of SOLEM-GR in
non-Clifford circuits and compare it with the standard ZNE heuristic.
They are performed directly at the logical level, so each simulated qubit
represents a logical qubit. We assume that the conditional logical errors
are independent across space and time, use the same set of conditional
logical channels at every qubit and layer, and neglect their dependence on
incoming logical errors.

We consider a Floquet circuit defined on an open \(4\times4\) square
lattice of \(n=16\) qubits.  One Floquet cycle consists of uniform
single-qubit \(X\)- and \(Z\)-rotation layers, followed by
nearest-neighbor fractional-angle \(ZZ\) rotations on the four
edge-colored bond layers. 
The evolution is generated by
\begin{align}
U_{F}
={}&
\left(\prod_{r=1}^{4}e^{-i\frac{\alpha_{zz}}{2} H_r }\right)
e^{-i \frac{\alpha_z}{2}  H_Z}
e^{-i\frac{\alpha_x}{2} H_X},
\label{eq:toy-floquet-step}\\
H_X
={}&\sum_{i=1}^{n}X_i,\qquad
H_Z=\sum_{i=1}^{n}Z_i,\qquad
H_r=\sum_{\langle i,j\rangle\in E_r}Z_iZ_j .
\nonumber
\end{align}
Here \(E_1,E_2,E_3,E_4\) are the horizontal-even, horizontal-odd,
vertical-even, and vertical-odd edge sets. 
The unitary \(U_{F}\)
therefore defines a Floquet mixed-field Ising model. We use
\(\alpha_{zz}=\pi/6\) and \(\alpha_x=1\), and variable \(\alpha_z\).

We initialize the state
\(\bigotimes_{i=1}^{n}\ket{0}_i\) and, after the final Floquet cycle,
measure the magnetization
\[
M =\frac{1}{n}\sum_{i=1}^{n} Z_i.
\]
After every circuit layer, each qubit independently acquires a syndrome
label sampled from the P2LC distribution and a Pauli error sampled from
the corresponding syndrome-conditioned logical channel. 
 
Throughout the circuit, we use the same syndrome distribution \(p_s\)
and the same syndrome-conditioned one-qubit Pauli channels
\(\Lambda_{L|s}\), obtained from the history-zero Steane-code P2LC.
There are therefore
\(V=nD\) locations at which one-qubit syndrome-conditioned Pauli channels
act, where \(D\) is the total number of circuit layers.

\textit{SOLEM-ZNE.}
We use the fine-grained syndrome-conditioned channels
\(\Lambda_{L|s}\).
For each amplification factor \(a\), we obtain the \emph{exact} coefficients \(q_s(a)\)
using the analogue of the optimization in
Eq.~\eqref{eq:qp_cost_func} to obtain the channel 
\(\Lambda(a)=a\Lambda_L+(1-a)I\). 
Here, the reference channel is the syndrome-averaged logical error channel, which takes the simple form
\[
\Lambda_L=\sum_s p_s\Lambda_{L|s}.
\]
We then
implement two-point SOLEM-ZNE using the resulting correlated estimators and the 
response function \( f_{\bst}(a) =\theta_1 \exp(-\theta_2 a)\).

\textit{SOLEM-GR.}
We consider two syndrome partitions.
The first is the fine-grained partition.
The second is a binary partition constructed such that the channels
conditioned on \(k=0\) and \(k=1\) are ``parallel'' in the sense of
Eq.~\eqref{eq:parallel}.
This binary partition realizes the channel structure assumed in our solvable
model, allowing us to test the applicability of its predictions.
For the regression, we use
\[
Z_{\phi}(\bsk)=\exp \left( -\phi \sum_{j=1}^{V}\epsilon_{L|k_j} \right),
\]
where we assume \( \mathbb{E}[o\mid Z_{\phi} =z ] \simeq \theta z\), 
with \(-1\leq\theta \leq1\) and \(\phi>0\). We fit the parameters by
least squares.

\begin{figure}[ht!]
    \begin{centering}
    \includegraphics[width=\columnwidth]{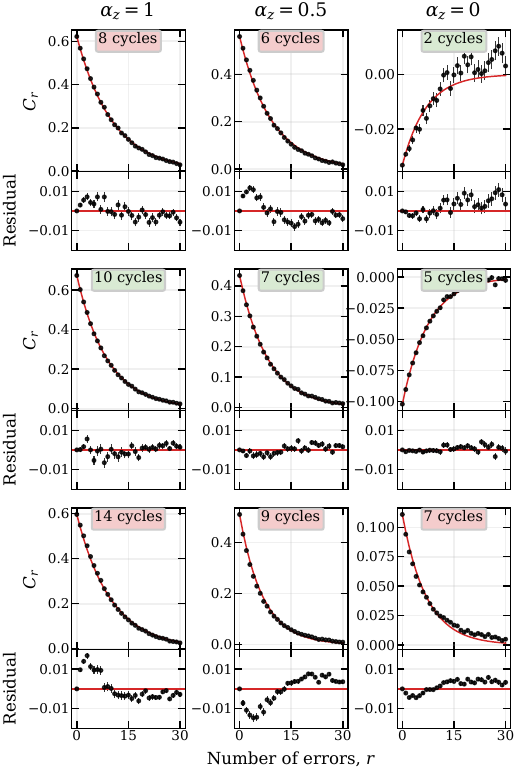}
    \end{centering}
    \caption{
    \(C_r\) coefficients for the Floquet circuits of
    Eq.~\eqref{eq:toy-floquet-step}, with error
    \(E=0.418X+0.176Y+0.406Z\) and observable
    \(n^{-1}\sum_{i=1}^{n}Z_i\).
    Here, \(E\) is the normalized non-identity part of the logical channels in
    the binary ``parallel'' partition.
    The columns correspond to different values of \(\alpha_z\).
    Green (red) shading indicates that the corresponding estimator in
    Fig.~\ref{fig:flouqet_sim} deviates from the ideal value by less (more)
    than \(2\sigma\).
    We estimate \(C_r\) by uniformly sampling sets of \(r\) error locations and
    drawing the Pauli error at each location according to \(E\).
    We then perform a weighted least-squares fit to
    \(C_r=\E{O}(1-\rho)^r\) to obtain \(\hat{\rho}\).
    The upper panels show the estimated \(C_r\) values and their exponential
    fits; the lower panels show the residuals
    \(C_r-\E{O}(1-\hat{\rho})^r\).
    For the selected circuits, statistically resolved deviations of SOLEM-GR
    with the binary partition occur together with systematic departures of
    \(C_r\) from the exponential model.
    }
    \label{fig:cr_in_floquet}
\end{figure}

Fig.~\ref{fig:flouqet_sim} shows the SOLEM-GR and SOLEM-ZNE results for 
various values of \(\alpha_z\). The deviations of the GR and ZNE estimators
from the ideal value are strongly correlated in this setup, as shown by the
residual panels.
SOLEM-GR with \(K=448\) consistently has lower overhead than SOLEM-ZNE.
As noted in Sec.~\ref{sec:model}, refining the ZNE construction may reduce
this gap; we do not pursue such refinements here.
For the \(K=2\) parallel partition, 
we estimate the coefficients \(C_r\)
introduced in Sec.~\ref{sec:model}, Prop.~\ref{prop:exp_Cr}. 
Interestingly, we find that for most cycle numbers the coefficients 
\(C_r\) are approximated very well by an exponential decay. 
Fig.~\ref{fig:cr_in_floquet} shows the corresponding exponential fits and
their residuals. For the selected cycle counts, statistically resolved
estimator deviations coincide with systematic departures of \(C_r\) from the
exponential response, consistent with Prop.~\ref{prop:exp_Cr}.

\clearpage
\onecolumngrid
\appendix

\section{Proofs and derivations for Sec.~\ref{sec:solem-qp-based} }

\subsection{Derivative estimation}\label{sec:deriv}
As discussed in Sec.~\ref{sec:solem-zne}, the same measured records
provide estimators of \(\E{O}_a\) at multiple values of \(a\). They can
therefore also be used to estimate derivatives with respect to \(a\).
For example, the forward-difference estimator
\begin{align}\label{eq:solem-zne-finite-difference}
\widehat{O}^{(l)}_{a_0,h}
:=
\frac{1}{h^l}\sum_{m=0}^{l}
(-1)^{l-m}\binom{l}{m}\widehat{O}_{a_0+mh}
\end{align}
converges to
\(\left.d^l\widehat{O}_a/da^l\right|_{a=a_0}\) as \(h\to0\).
Importantly, all terms in Eq.~\eqref{eq:solem-zne-finite-difference}
are evaluated on the same shots and are therefore correlated. The
derivative can equivalently be obtained directly, without introducing a
finite-difference step size, by differentiating the shot weights. Because
the point estimators use the same records, the finite-difference
combination converges shot by shot to this direct estimator.

Define
\begin{align}
d_{\bsk}^{(j)}:=y_{\bsk}^{(j)}-q_{\bsk}^{(j)},
\qquad
y_{\bsk}^{(j)}(a)=q_{\bsk}^{(j)}+a d_{\bsk}^{(j)}.
\end{align}
The global coefficient in Eq.~\eqref{eq:solem-zne-estimator} is the
degree-\(D\) polynomial
\begin{align}
y_{\bsk}(a)
=\prod_{j=1}^{D}\left(q_{\bsk}^{(j)}+a d_{\bsk}^{(j)}\right).
\end{align}
Its normalized \(l\)th derivative is
\begin{align}\label{eq:solem-zne-weight-derivative}
\frac{1}{l!}\frac{d^l y_{\bsk}(a)}{da^l}
=
\sum_{1\leq j_1<\cdots<j_l\leq D}
\prod_{j\in\{j_1,\ldots,j_l\}}d_{\bsk}^{(j)}
\prod_{j\notin\{j_1,\ldots,j_l\}}y_{\bsk}^{(j)}(a).
\end{align}
Consequently,
\begin{align}\label{eq:def_Dl}
\mathsf D_l(a_0)
:={}&
\frac{1}{N}\sum_{i=1}^{N}
\frac{o_i}{p_{\bsk_i}}
\left.
\frac{1}{l!}\frac{d^l y_{\bsk_i}(a)}{da^l}
\right|_{a=a_0}
\end{align}
is unbiased for the normalized derivative
\begin{align}
\mathbb E[\mathsf D_l(a_0)]
=
\left.
\frac{1}{l!}\frac{d^l\E{O}_a}{da^l}
\right|_{a=a_0}.
\end{align}
Equation~\eqref{eq:solem-zne-weight-derivative} need not be evaluated
by explicitly summing over the \(\binom{D}{l}\) subsets. The
coefficients through order \(r\) in
\(y_{\bsk}(a_0+t)=\prod_j[y_{\bsk}^{(j)}(a_0)+t d_{\bsk}^{(j)}]\)
can be accumulated by truncated polynomial multiplication in
\(O(Dr)\) time and \(O(r)\) memory per shot.

\subsection{Shared-shot GLS}
\label{sec:zne-shared-shot-gls}

\begin{observation}\label{obs:shared-shot-gls}
Suppose that the noise response belongs to the fitted family,
\(\E{O}_a=\phi_{\boldsymbol\theta_*}(a)\) for all admissible \(a\).
Choose amplification points \(a_1,\ldots,a_m\) that identify
\(\boldsymbol\theta_*\), and assume the usual regularity conditions for
consistent, asymptotically normal GLS. Both shared-shot and independently
allocated GLS then converge to \(\boldsymbol\theta_*\); only their
asymptotic covariance differs. Let \(J\) be the Jacobian at
\(\boldsymbol\theta_*\), and let
\(\Omega_{\mathrm{sh}}\) be the asymptotic covariance coefficient when
every shot is used at every point. For an independent allocation with
fractions \(w_r>0\), \(\sum_r w_r=1\), define
\(\Omega_{\mathrm{alloc}}(\boldsymbol w):=
\operatorname{diag}((\Omega_{\mathrm{sh}})_{rr}/w_r)_{r=1}^m\).
If \(\Omega_{\mathrm{sh}}\) is positive definite and \(J\) has full
column rank, the asymptotic GLS covariance coefficients obey
\[
\left(J^{\mathsf T}\Omega_{\mathrm{sh}}^{-1}J\right)^{-1}
\preceq
\left[
J^{\mathsf T}\Omega_{\mathrm{alloc}}(\boldsymbol w)^{-1}J
\right]^{-1}.
\]
Thus, shared-shot GLS has no larger asymptotic variance in any
parameter direction than any independent shot allocation.
\end{observation}
\begin{proof}
Let \(\boldsymbol e_r\) be the \(r\)th standard basis vector.
Cauchy--Schwarz gives
\(\Omega_{\mathrm{sh}}^{-1}\succeq
\boldsymbol e_r\boldsymbol e_r^{\mathsf T}/
(\Omega_{\mathrm{sh}})_{rr}\) for every \(r\). Taking the weighted
average of these inequalities yields
\[
\Omega_{\mathrm{sh}}^{-1}
\succeq
\sum_{r=1}^{m}
w_r\frac{\boldsymbol e_r\boldsymbol e_r^{\mathsf T}}
{(\Omega_{\mathrm{sh}})_{rr}}
=\Omega_{\mathrm{alloc}}(\boldsymbol w)^{-1}.
\]
Congruence with \(J\) and inversion, which reverses the
positive-semidefinite order, give the claimed inequality.
\end{proof}

\section{Derivations and proofs for the analytical model}\label{app:model}
\subsection{Proof of Prop.~\ref{prop:ev-cond-N1}} \label{app:proof-ev-cond-N1}
\begin{proof}
Let \(\mathcal J_E\subseteq[V]\) denote the random set of locations at
which \(E\) occurs. Conditioned on \(N_E=r\) and \(N_1=l\), exchangeability
of the syndrome-and-error process implies that \(\mathcal J_E\) is uniformly
distributed over the subsets of \([V]\) of size \(r\). Hence, for every
\(J\subseteq[V]\) with \(|J|=r\),
\[
\mathbb P(\mathcal J_E=J\mid N_E=r,N_1=l)
=\frac{1}{\binom Vr}.
\]
Therefore,
\begin{align*}
\mathbb E[o\mid N_E=r,N_1=l]
&=\sum_{\substack{J\subseteq[V]\\|J|=r}}
\mathbb P(\mathcal J_E=J\mid N_E=r,N_1=l)\,o_J
\\
&=\frac{1}{\binom Vr}
\sum_{\substack{J\subseteq[V]\\|J|=r}}o_J
=C_r.
\end{align*}
The law of total expectation then gives
\begin{align*}
\mathbb E[o\mid N_1=l]
&=\sum_{r=0}^{V}
\mathbb P(N_E=r\mid N_1=l)\,
\mathbb E[o\mid N_E=r,N_1=l]
\\
&=\sum_{r=0}^{V}K_{r,l}C_r,
\end{align*}
where \(K_{r,l}:=\mathbb P(N_E=r\mid N_1=l)\).
\end{proof}

For later use, 
the probability-generating function is
\begin{align}\label{eq:K-gen-func}
  \mathbb E[z^{N_E}\mid N_1=l]
  =\sum_{r=0}^{V}K_{r,l}z^r
  =(1-\epsilon_1+\epsilon_1z)^l
   (1-\epsilon_0+\epsilon_0z)^{V-l}.
\end{align}

\subsection{Proof of Prop.~\ref{prop:exp_Cr}}\label{app:proof_exp_cr}
\begin{proof}
\emph{(i) \(\implies\) (ii).}
Setting \(z=1-\rho\) in Eq.~\eqref{eq:K-gen-func} gives
\[
\mathbb E[o\mid N_1=l]
=C_0\sum_{r=0}^{V}K_{r,l}(1-\rho)^r
=C_0(1-\rho\epsilon_1)^l(1-\rho\epsilon_0)^{V-l}.
\]

\emph{(ii) \(\implies\) (iii).}
For fixed \(\epsilon_0\) and \(\epsilon_1\), the conditional means in
condition~(ii) do not depend on \(p_1\).  Hence, for every
\(p_1\in(0,1)\), averaging over
\(N_1\sim\operatorname{Bin}(V,p_1)\) gives
\[
\mathbb E[o]
=C_0\left[p_1(1-\rho\epsilon_1)
 +(1-p_1)(1-\rho\epsilon_0)\right]^V
=C_0(1-\rho\epsilon_L)^V.
\]
As \(p_1\) varies in \((0,1)\), \(\epsilon_L\) varies over the open
interval \((\epsilon_0,\epsilon_1)\), proving condition~(iii).

\emph{(iii) \(\implies\) (i).}
On the assumed open interval,
\begin{align*}
\sum_{r=0}^{V}\binom{V}{r}\epsilon_L^r
 (1-\epsilon_L)^{V-r}C_r
&=C_0(1-\rho\epsilon_L)^V\\
&=C_0\sum_{r=0}^{V}\binom{V}{r}\epsilon_L^r
 (1-\epsilon_L)^{V-r}(1-\rho)^r .
\end{align*}
Dividing by \((1-\epsilon_L)^V\), setting
\(x=\epsilon_L/(1-\epsilon_L)\), and equating coefficients of \(x^r\)
gives \(C_r=C_0(1-\rho)^r\).
\end{proof}

\subsection{Overhead of two-point \gls{mname}-ZNE}\label{sec:app_solem_zne}
For a shot containing \(l\) syndrome-1 outcomes, the reweighting factor is
\(
  w_a(l)
  =\left(\frac{q_1(a)}{p_1}\right)^l
   \left(\frac{q_0(a)}{1-p_1}\right)^{V-l}
\).
Thus, \(
  \widehat o_a=\frac{1}{N}\sum_{i=1}^N o_iw_a(|\bsk_{i}|), \,
  \E{O}_a:=\mathbb E[o\,w_a(N_1)].
\)
Assuming \(o_i\in\{-1,1\}\), independence of the
shots gives
\begin{equation}
  \operatorname{Cov}(\widehat o_{a_1},\widehat o_{a_2})
  =\frac{1}{N}\left(
    \mathbb E[w_{a_1}(N_1)w_{a_2}(N_1)]-\E{O}_{a_1}\E{O}_{a_2}
  \right).
  \label{eq:amplified-estimator-covariance-general}
\end{equation}
Consequently, the covariance has the explicit form
\begin{align}
  \operatorname{Cov}(\widehat o_{a_1},\widehat o_{a_2})
  =\frac{1}{N}\Bigg\{&
  \left[
    \frac{\epsilon_0^2(a_1-1)(a_2-1)}
         {p_1(\epsilon_1-\epsilon_0)^2}
    +\frac{(\epsilon_1-a_1\epsilon_0)
            (\epsilon_1-a_2\epsilon_0)}
           {(1-p_1)(\epsilon_1-\epsilon_0)^2}
  \right]^V
  -\E{O}_{a_1}\E{O}_{a_2}\Bigg\}.
  \label{eq:amplified-estimator-covariance-explicit}
\end{align}

Now we can use generalized least-squares for the exact model
\(\phi_{\boldsymbol{\theta}}(a)
=\theta_1(1-\theta_2a\epsilon_0)^V\),
with true parameter values
\(\theta_1=\E{O}\) and \(\theta_2=\rho\).
Let
\(
\Omega_{j,k} := N \operatorname{Cov}(\widehat o_{a_j},\widehat o_{a_k})\), 
so that by definition \(\det \Omega \geq 0\).
We have
\(
  J_{j,1}=(1-\theta_2a_j\epsilon_0)^V,
  \,
  J_{j,2}=-\theta_1Va_j\epsilon_0
  (1-\theta_2a_j\epsilon_0)^{V-1},
\)
and therefore, for two distinct amplification factors \(0< a_1<a_2\),
assuming \(0<\rho a_i\epsilon_0<1\) and \(\det \Omega > 0\), we obtain the asymptotic variance as \(N \rightarrow \infty\)
\begin{align}
  N\V(\widehat{\theta}_1)
  &= 
  \left[
    \left(J^t \Omega^{-1}J\right)^{-1}
  \right]_{11}
  =
  \left[J^{-1} \Omega (J^{-1})^t\right]_{11}
  \notag\\
  &=\frac{1}{(a_2-a_1)^2}  
  \Bigg[
  \left(
  \frac{a_2\sqrt{\Omega_{11}}}{(1-\rho a_1\epsilon_0)^{V-1}}
  +\frac{a_1\sqrt{\Omega_{22}}}{(1-\rho a_2\epsilon_0)^{V-1}}
  \right)^2
  -\frac{2a_1a_2
  \left(\sqrt{\Omega_{11}\Omega_{22}}+\Omega_{12}\right)}
  {[(1-\rho a_1\epsilon_0)(1-\rho a_2\epsilon_0)]^{V-1}}
  \Bigg] 
  \label{eq:theta-zero-variance-explicit}. 
\end{align}
The first squared term, including the prefactor \(1/(a_2-a_1)^2\), 
is the minimum overhead obtainable by estimating the two amplified 
expectation values independently and optimally allocating a total of \(N\) 
shots between them.
The second term is non-negative, 
so that the correlation between \(\hat{o}_{a_1}\) and \(\hat{o}_{a_2}\)
can only reduce the estimation variance. This is a special case of Obs.~\ref{obs:shared-shot-gls}.

We will assume that \(\epsilon_0\ll\epsilon_1\) and \(p_1\ll1\).
We obtain an approximate location of the optimum in two steps.
First, we use the variance for optimal shot allocation, 
remove the negative terms \(-\E{O}_{a_i}^2\) in \(\Omega_{ii}\), 
and approximate \( (1-\rho a \epsilon_0)^{V-1} \) as 
\(\exp(-\rho a \epsilon_0V)\)
to get
\begin{align}
&\frac{1}{(a_2-a_1)^2}
\left[
a_2\exp\!\left(
\rho V\epsilon_0a_1+
\frac{V(\epsilon_L-a_1\epsilon_0)^2}
{2p_1\epsilon_1^2}
\right)
\;+
a_1\exp\!\left(
\rho V\epsilon_0a_2+
\frac{V(\epsilon_L-a_2\epsilon_0)^2}
{2p_1\epsilon_1^2}
\right)
\right]^2
\end{align}
and the interior optimum
\begin{equation}
  a_{1,2}^{\star}
  =a_{EC} (\bar{\lambda}-1) \mp \frac{\epsilon_1}{\epsilon_0}\sqrt{\frac{p_1}{V}}
\end{equation}
where here \(a_{EC}=\epsilon_L / \epsilon_0\), and 
\begin{align}
\bar{\lambda}:=2-\rho\epsilon_1p_{1|L}.
\end{align}
The preceding calculation identifies the center of the two
amplification points.  Set
\begin{equation}
a_c:=a_{EC}(\bar{\lambda}-1),
\qquad
a_1=a_c+x_1,
\qquad
a_2=a_c+x_2,
\end{equation}
where \(x_1\ne x_2\).
Expanding the exact asymptotic variance in Eq.~\eqref{eq:theta-zero-variance-explicit} around
\(a_1=a_2=a_c\), and using
\(0<\epsilon_0\ll\epsilon_1<1\) and \(p_1\ll1\), gives
\begin{align*}
N\V(\widehat{\theta}_1)
={}&
e^{\bar{\lambda}\rho V\epsilon_L}
\left[
1+\frac{V\epsilon_L}{p_{1|L}\epsilon_1}
(\bar{\lambda}-1)^2
\right]
-\E{O}^2
\\
&\quad+O(x_1^2+x_2^2)
\\
&\quad+O\!\left(
e^{\bar{\lambda}\rho V\epsilon_L}
\left[
1+\frac{V\epsilon_L}{p_{1|L}\epsilon_1}
(\bar{\lambda}-1)^2
\right]
\left[
p_1+\frac{\epsilon_0}{\epsilon_1}
+V\left(p_1+\frac{\epsilon_0}{\epsilon_1}\right)^2
\right]
\right).
\end{align*}
Thus, the full correlated variance has a finite limit as \(a_1,a_2\rightarrow a_c, \, a_1\ne a_2 \).

\subsection{Fisher-information limit}\label{sec:fisher}
To identify the smallest variance allowed by the parametric model,
consider instead its exact conditional Bernoulli likelihood. 
Let
\begin{equation}
Z_\phi(l):=(1-\phi\epsilon_0)^{V-l}
(1-\phi\epsilon_1)^l.
\end{equation}
This is the global decay factor of the solvable model, with noiseless
continuation \(Z_\phi^{(0)}=1\). The SOLEM-GR response is therefore
\(\theta Z_\phi(l)\), where \(\theta\) is the ideal expectation value and
\(\phi\) is a nuisance parameter.
For \(o\in\{-1,1\}\) and \(z\in\{0,\ldots,V\}\), the joint distribution is
\[
 p(o,z; \theta,\phi)
=\binom{V}{z}p_1^z(1-p_1)^{V-z}
 \frac{1+o\theta Z_\phi(z)}{2},
\]
and its
maximum-likelihood estimator minimizes the negative log-likelihood
\begin{equation}\label{eq:solem-gr-maximum-likelihood-loss}
 (\widehat\theta_{\rm ML},\widehat\phi_{\rm ML})
 =\argmin_{\theta,\phi}
 \left\{-\sum_{i=1}^{N}
 \log\!\left[\frac{1+o_i\theta Z_\phi(z_i)}{2}\right]\right\}.
\end{equation}
The binomial probability of \(z_i\) is omitted because it is independent of
\((\theta,\phi)\).

At the true parameter values \(\theta=\E{O}\) and \(\phi=\rho\), define, for
\(k=0,1,2\),
\begin{equation}\label{eq:solem-gr-fisher-Sk}
 S_k:=\sum_{l=0}^{V}\binom{V}{l}p_1^l(1-p_1)^{V-l}
 \frac{Z_\rho(l)^2}{1-\E{O}^2Z_\rho(l)^2}
 \left[
 \frac{(V-l)\epsilon_0}{1-\rho\epsilon_0}
 +\frac{l\epsilon_1}{1-\rho\epsilon_1}
 \right]^k.
\end{equation}
The Fisher information per shot, evaluated at the true value, is then
\begin{equation}\label{eq:solem-gr-fisher-matrix}
 I(\E{O},\rho)
 =\mathbb{E}_{\E{O},\rho}
 \!\left[
 \partial_{(\theta,\phi)}\log p(o,z;\theta,\phi)
 \partial_{(\theta,\phi)}\log p(o,z;\theta,\phi)^t
 \right]_{(\theta,\phi)=(\E{O},\rho)}
 =\begin{pmatrix}
 S_0&-\E{O}S_1\\
 -\E{O}S_1&\E{O}^2S_2
 \end{pmatrix}.
\end{equation}
This gives the Cram\'er--Rao bound
\begin{equation}\label{eq:solem-gr-fisher-bound}
 \liminf_{N\to\infty}
 N\V_{\E{O},\rho}(\widehat{\theta}_{N})
 \geq \Gamma_{\mathrm F}
 :=\left[I(\E{O},\rho)^{-1}\right]_{11}
 =\left(S_0-\frac{S_1^2}{S_2}\right)^{-1}.
\end{equation}
Within this parametric family, this is a bound on every regular,
asymptotically unbiased estimator based on the coarsened observations
\((z_i,o_i)=(|\bsk_i|,o_i)\).  
Under the usual regularity and identifiability conditions, the
maximum-likelihood estimator in
Eq.~\eqref{eq:solem-gr-maximum-likelihood-loss} asymptotically attains it.

For \(0<\epsilon_0\ll\epsilon_1<1\) and \(p_1\ll1\), with
\(V\epsilon_0\) and \(Vp_1\) kept finite, we expand the exact binomial
expressions in Eq.~\eqref{eq:solem-gr-fisher-Sk} in \(p_1\) and
\(\epsilon_0/\epsilon_1\). We also retain the first nonvanishing contribution
of the squared conditional mean, using
\[
\frac{1}{1-\E{O}^2Z_\rho(l)^2}
=1+\E{O}^2Z_\rho(l)^2
+O\!\left(
\left[\max_{0\leq l\leq V}|\E{O}Z_\rho(l)|\right]^4
\right).
\]
Substitution into Eq.~\eqref{eq:solem-gr-fisher-bound} gives
\begin{align}
\Gamma_{\mathrm F}
={}&
\exp\!\left(\bar{\lambda}\rho V\epsilon_L\right)
\left[1+\frac{V\epsilon_L}{p_{1|L}\epsilon_1}
(\bar{\lambda}-1)^2\right]
-\E{O}^2 
\exp\!\left[
(2-\bar{\lambda})\rho V\epsilon_L
(2-\rho\epsilon_1)^2
\right]F_0
\nonumber\\
&\quad+O\!\Bigg(
e^{\bar{\lambda}\rho V\epsilon_L}
\left[1+\frac{V\epsilon_L}{p_{1|L}\epsilon_1}
(\bar{\lambda}-1)^2+F_0\right]
\left[
p_1+\frac{\epsilon_0}{\epsilon_1}
+V\left(p_1+\frac{\epsilon_0}{\epsilon_1}\right)^2
\right]\Bigg)
\nonumber\\
&\quad+O\!\left(
e^{\bar{\lambda}\rho V\epsilon_L}
\left[1+\frac{V\epsilon_L}{p_{1|L}\epsilon_1}
(\bar{\lambda}-1)^2\right]
\left(\max_{0\leq l\leq V}\left|\E{O}Z_\rho(l)\right| \right)^4
\right)
,
\label{eq:solem-gr-fisher-bound-with-mean}
\end{align}
where
\begin{equation}\label{eq:gr-F0}
F_0:=
\left[1+\rho V\epsilon_L(\bar{\lambda}-1)
(1-\rho\epsilon_1)(2-\rho\epsilon_1)\right]^2
+\frac{V\epsilon_L}{p_{1|L}\epsilon_1}(\bar{\lambda}-1)^2
(1-\rho\epsilon_1)^2.
\end{equation}

\subsection{SOLEM-GR with squared loss}
\label{sec:app_solem_gr}

Consider the unweighted nonlinear least-squares estimator
\[
(\widehat\theta_{\mathrm{LS}},\widehat\phi_{\mathrm{LS}})
=
\arg\min_{\theta,\phi}
\frac{1}{N}\sum_{i=1}^{N}
\left[o_i-\theta Z_\phi(z_i)\right]^2,
\qquad z_i=|\bsk_i|.
\]
Under the standard regularity and moment conditions for nonlinear least
squares~\cite{Jennrich1969,Hansen2022Econometrics}, its asymptotic covariance
is \(N^{-1}H^{-1}GH^{-1}\), where
\[
H=\mathbb E[g(z)g(z)^t],
\qquad
G=\mathbb E[\V(o\mid z)g(z)g(z)^t],
\qquad
g(z)=
\left.\partial_{(\theta,\phi)}[\theta Z_\phi(z)]
\right|_{(\theta,\phi)=(\E{O},\rho)}.
\]

For \(o\in\{-1,1\}\), substituting
\(\V(o\mid z)=1-\E{O}^2Z_\rho(z)^2\) and expanding under the same assumptions
used for the Fisher bound gives
\[
0\leq
\lim_{N\to\infty}
N\V_{\E{O},\rho}\!\left(\widehat{\theta}_{N}^{\mathrm{LS}}\right)
-\Gamma_{\mathrm F}
=
O\!\left(
\left[
\max_{0\leq l\leq V}
\left|
\E{O}(1-\rho\epsilon_0)^{V-l}
(1-\rho\epsilon_1)^l
\right|
\right]^4
\right).
\]
Thus, unweighted least squares agrees with the Fisher-information bound
through quadratic order in the conditional mean.

\bibliography{references.bib}

\end{document}